\documentclass[a4paper,UKenglish,cleveref, autoref, thm-restate]{lipics-v2021-arxiv}

\newtheorem*{open-problem}{Open Problem}

\theoremstyle{definition}
\newtheorem{algorithm-stmt}[algorithm]{Algorithm}

\def\defn#1{\textit{\textbf{\boldmath #1}}\index{#1}}

\graphicspath{{graphics/}}

\newcommand{\changed}[1]{{#1}}

\title{Morphing Planar Graph Drawings via Orthogonal Box Drawings}

\author{Therese Biedl} {University of Waterloo, Waterloo, Ontario, Canada} {biedl@uwaterloo.ca}{
0000-0002-9003-3783
}{Research supported by NSERC FRN RGPIN-2020-03958}
\author{Anna Lubiw} {University of Waterloo, Waterloo, Ontario, Canada} {alubiw@uwaterloo.ca}{0000-0002-2338-361X}{Research supported by NSERC}
\author{Jack Spalding-Jamieson} {University of Waterloo, Waterloo, Ontario, Canada} {jacketsj@gmail.com}{0000-0002-1209-4345}{}

\authorrunning{T. Biedl and A. Lubiw and J. Spalding-Jamieson} %

\Copyright{Therese Biedl and Anna Lubiw and Jack Spalding-Jamieson} %


\category{} %

\relatedversion{
\changed{This work has appeared as the third author's Master's thesis.}
}
\relatedversiondetails{Thesis url}{http://hdl.handle.net/10012/20172 }

\nolinenumbers %

\begin{document}

\maketitle

\begin{abstract}

We give an algorithm to morph planar graph drawings that achieves small grid size at the expense of allowing a constant number of bends on each edge.
The input is an $n$-vertex planar graph 
and two planar straight-line drawings 
of the graph on an $O(n) \times O(n)$ grid.   The planarity-preserving morph is composed of $O(n)$ linear morphs between successive pairs of  
drawings, each on an  $O(n) \times O(n)$ grid with a constant number of bends per edge.
The algorithm to compute the morph runs in $O(n^2)$ time on a word RAM model with standard arithmetic operations---in particular 
no square roots or cube roots are required.

The first step of the algorithm is to morph each input drawing to a planar orthogonal box drawing where vertices are represented by boxes and each edge is drawn as a horizontal or vertical segment.  
The second step is to morph between planar orthogonal box drawings.  This is done  
by extending known techniques for morphing planar orthogonal drawings with point vertices.
\end{abstract}

\section{Introduction}
\label{sec:introduction}

Algorithms to compute a straight-line drawing of a planar $n$-vertex graph on an $O(n) \times O(n)$ grid have been known since the 1980's~\cite{de1990draw, schnyder1990embedding}, but it is an open problem to 
achieve such straight-line small-grid results for planar graph morphing.  
To make this precise, 
let $P$ and $Q$ be
two planar straight-line drawings of a graph $G$ that are \defn{compatible}, meaning that they have the same faces and the same outer face. 
A \defn{linear morph sequence} from $P$ to $Q$ is a sequence of \defn{explicit intermediate drawings} of $G$ starting with $P$ and ending with $Q$.  By taking a \defn{linear morph}, i.e.,  a %
linear interpolation of 
vertex positions,
between each successive pair of explicit intermediate drawings, we obtain a continuous \defn{piece-wise linear morph} from $P$ to $Q$ indexed by time $t \in [0,1]$.  
The morph is \defn{planarity-preserving} if the drawing at every time $t \in [0,1]$ is planar.
A straight-line drawing of a graph \defn{lies on a grid} if the points representing the vertices lie at grid points; in case edges are drawn as poly-lines with bends, the bends must also lie at grid points. 
The following problem is open:

\begin{open-problem}
For a planar graph $G$ with $n$ vertices and a compatible pair of planar straight-line drawings $P$ and $Q$ of $G$ on an $O(n)\times O(n)$ grid,
is there a planarity-preserving piece-wise linear morph 
from $P$ to $Q$ 
where each explicit intermediate drawing is a straight-line drawing on an $O(n)\times O(n)$ grid?
\end{open-problem}

The morphing algorithm of Alamdari, Angelini, Barrera-Cruz, Chan, Da Lozzo, Di Battista, Frati, Haxell, Lubiw, Patrignani, Roselli, Singla, and Wilkinson~\cite{alamdari2017morph}---which is based on Cairns' edge contraction method---finds a planarity-preserving linear morph sequence of length $O(n)$ but with no guarantee on the grid size of the explicit intermediate drawings.  In fact, the vertex coordinates are computed using cube roots on a real RAM model, and vertices may become almost coincident to imitate edge contractions.
The recent morphing algorithm of Erickson and Lin~\cite{erickson2023toroidal}
achieves the same bound of $O(n)$ linear morphs 
for the subclass of 3-connected graphs and avoids the edge-contraction paradigm by following Floater's method of interpolating the matrix of barycentric coordinates/weights.  However, computing the barycentric weights requires square roots, and, even %
if that were avoided,
the reliance on 
Tutte/Floater drawings means that the vertex coordinates can require $\Omega(n)$ bits of precision\changed{, and thus a grid of size $\Omega(2^n \times 2^n)$}~\cite{di2021tutte}.  
The open problem was solved for the special case of Schnyder drawings by Barrera-Cruz, Haxell, and Lubiw~\cite{barrera2019morphing}. 

\changed{Note that the Open Problem asks about the \emph{existence} of a small-grid morph, regardless of algorithmic considerations.
Of course, one would also like a fast algorithm to find the morph, and current research has been much more successful with regard to run-time than grid-size. 
The general morphing algorithm of Alamdari et al.~\cite{alamdari2017morph} runs in  $O(n^3)$ time.
This was improved to $O(n^2 \log n)$, and to $O(n^2)$ for 2-connected graphs, by Klemz~\cite{klemz2021convex}. 
The morphing algorithm of Erickson and Lin~\cite{erickson2023toroidal},
which involves solving linear systems and only applies to 3-connected graphs, runs in $O(n^{1 + \omega/2})$ time, where $\omega < 2.371552$ is the matrix multiplication exponent~\cite{matrix-mult-2024}.
Note that a runtime of $O(n^2)$ is optimal if all the explicit intermediate drawings must be given as output, since $\Omega(n)$ such drawings may be required~\cite{alamdari2017morph}.}

We approach the 
\changed{Open Problem}
by insisting that explicit intermediate drawings lie on an $O(n) \times O(n)$ grid, but we relax the condition that edges be drawn as straight line segments, and allow a constant number of bends per edge.
Our main result is as follows.

\begin{restatable}
{theorem}{statethmmain}
\label{thm:main}
Let $G$ be a connected planar graph with $n$ vertices.
For a compatible pair of planar straight-line drawings $P$ and $Q$ of $G$
on an $O(n)\times O(n)$ grid,
there exists a planarity-preserving linear morph sequence
from $P$ to $Q$ of length $O(n)$,
where each explicit intermediate drawing 
lies on an $O(n)\times O(n)$ grid
and has $O(1)$ bends per edge.
Moreover, 
this sequence can be found in $O(n^2)$ time.
\end{restatable}

This result improves the algorithm of Lubiw and Petrick~\cite{lubiw2011morphing} that finds a linear morph sequence of length 
$O(n^6)$ with explicit intermediate drawings on an $O(n^3)\times O(n^3)$ grid with 
$O(n^5)$ bends per edge
and has run-time $O(n^6)$.
Besides the bad bounds, the algorithm introduces bends at non-grid points, which our present result avoids. %

For Theorem~\ref{thm:main}, the way we keep vertices of the  explicit intermediate drawings on the grid is to impose an even stronger condition that every edge is drawn as a poly-line with all segments along grid lines (in particular, orthogonal) except for the two segments at its endpoints.  The non-orthogonal segments incident to a vertex will live in an orthogonal rectangle called a \defn{box}.
Expressed differently, 
the first step of the algorithm is to morph each of the input straight-line drawings to a \defn{planar orthogonal box drawing} where a vertex is represented by a box, and an edge is drawn as an \defn{orthogonal poly-line}, a sequence of horizontal and vertical segments, joining two vertex boxes. 
See Figure~\ref{fig:overview-alg1}. 
In this way, we reduce our problem to the problem of morphing planar orthogonal box drawings.
In fact, we only need the special case where the 
input
orthogonal box drawings have no bends.

Our second main result is an algorithm to 
find a planarity-preserving morph between two compatible planar orthogonal box drawings. 
For a precise definition of such a morph, see Section~\ref{sec:preliminaries}. 
When the drawings lie on a small grid with few bends, 
our analysis of grid size, bends, and run-time allows us to prove Theorem~\ref{thm:main}.

\begin{theorem}
\label{thm:box-morph}
Let $G$ be a connected planar graph with $n$ vertices.  If $P$ and $Q$ are a compatible pair of planar orthogonal box drawings of $G$ on an $O(n) \times O(n)$ grid with $O(1)$ bends per edge, then there exists a planarity-preserving linear morph sequence from $P$ to $Q$ of length $O(n)$ where each explicit intermediate drawing 
is an orthogonal box drawing that
lies on an $O(n) \times O(n)$ grid with $O(1)$ bends. 
Moreover, this sequence can be found in $O(n^2)$ time.
\end{theorem}

\begin{figure}[hp]
\centering
\includegraphics[page=11,scale=0.47]{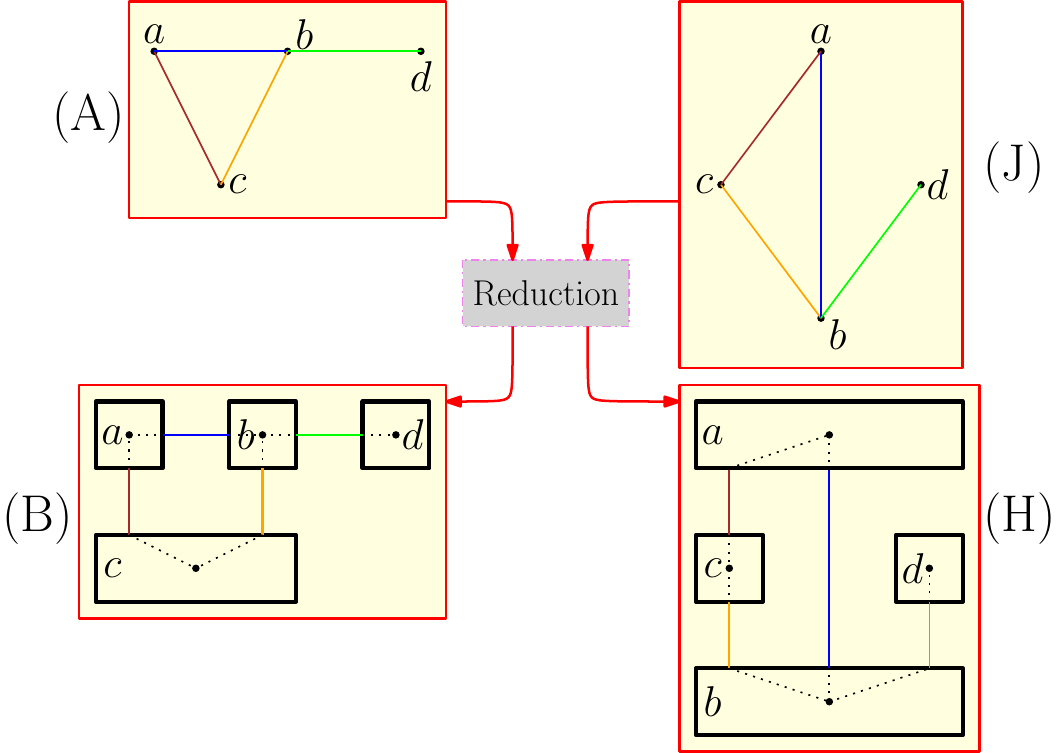}
\caption{
The two straight-line point drawings (A) and (J) are (respectively) morphed to admitted drawings of the orthogonal box drawings (B) and (H), by introducing two bends to each edge.
A morph from (B) to (H) (see the relevant %
part of Figure~\ref{fig:overview-alg2}) then induces a morph from (A) to (J). %
}
\label{fig:overview-alg1}
\vspace*{5mm}
\includegraphics[page=9,scale=0.47]{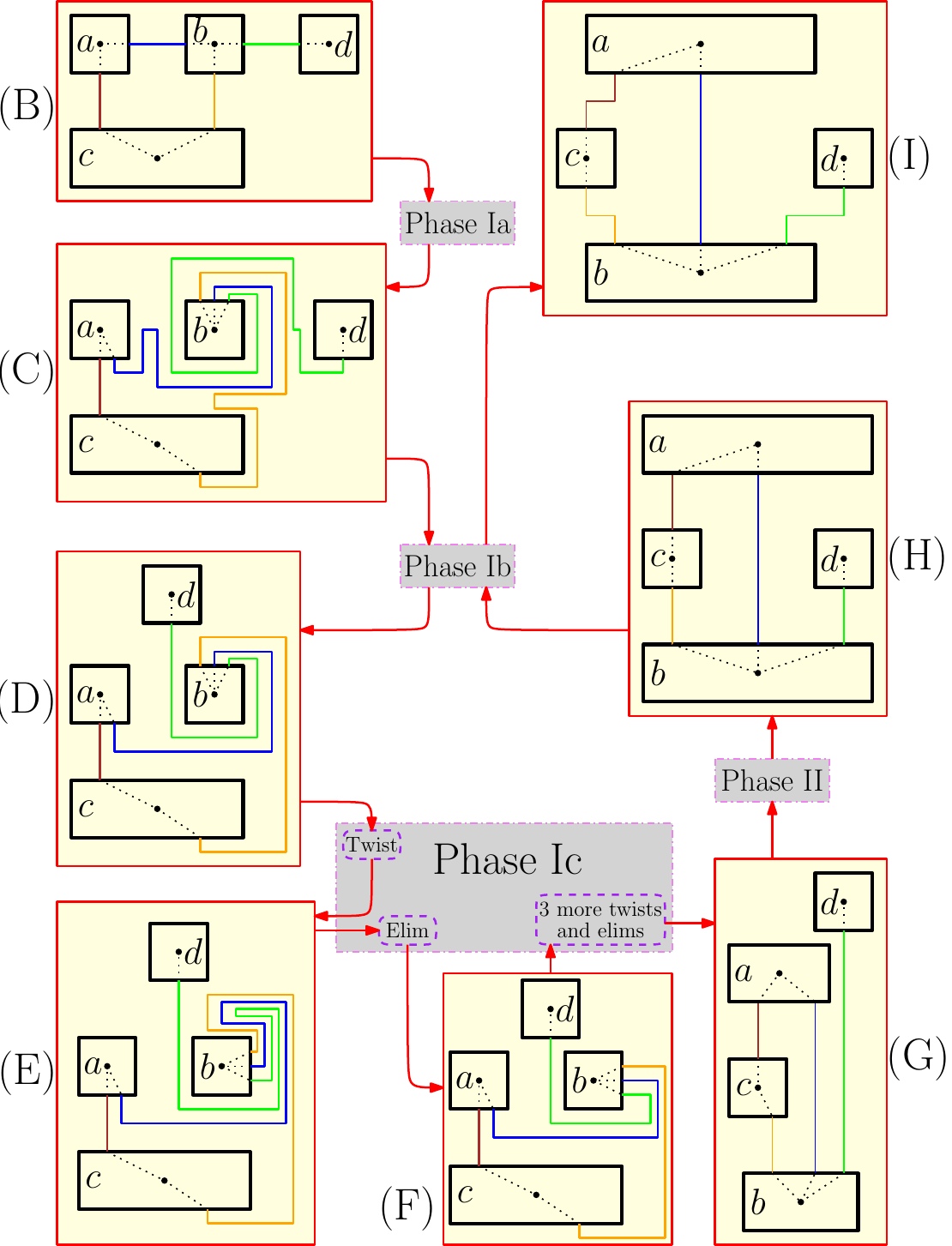}
\caption{Morphing between orthogonal box drawings where one drawing has bends. 
The input drawings are (B) and (I).
Phase Ia morphs (B) to (C), port-aligned with %
(I).
Phase Ib morphs (C) to (D) and (I) to (H), eliminating zig-zags.
In (D) we need 4 clockwise twists of $b$ 
to match spirality with (H). Twisting once morphs (D) to (E), and eliminating zig-zags gives (F).  
Three more twists and zig-zag eliminations give (G), which is morphed to the parallel drawing (H) in Phase II.    
}
\label{fig:overview-alg2}
\end{figure}

We prove Theorem~\ref{thm:box-morph} by building on algorithms to morph between planar orthogonal drawings where a vertex is drawn as a  point and an edge is drawn as an orthogonal poly-line, a more limited setting since vertices must have degree at most 4. For these algorithms, $n$ is the number of vertices plus bends. 
The algorithm of 
Biedl, Lubiw, Petrick, and Spriggs~\cite{biedl2013morphing} uses $O(n^2)$ linear morphs and runs in $O(n^3)$ time.
Van Goethem, Speckmann, and Verbeek~\cite{van2022optimal} improved the number of linear morphs to $O(n)$ by performing operations simultaneously, however without a run-time analysis.
We follow the two phases of Biedl et al.: (1) morph so that for each directed edge, its sequence of directions of edge segments (ignoring segment lengths) is the same in both drawings;
and (2) morph between such ``parallel'' orthogonal graph drawings.  We use the second phase as-is.  However, we must do some work to extend the first phase to the setting of orthogonal box drawings, and to improve the number of linear morphs to $O(n)$
by performing many operations simultaneously as in Van Goethem et al., while keeping the grid size small and the run-time fast.

\subsection{Preliminaries}
\label{sec:preliminaries}

Most graph drawing methods represent the vertices as points.
In a \defn{straight line} drawing an edge is drawn as a line segment between the vertex points, whereas in a \defn{poly-line} drawing an edge is drawn as a simple poly-line---a non-self-intersecting path of line segments joined at \defn{bends}.

In a \defn{linear morph} from drawing $P$ to $Q$, the position of a \defn{defining point} representing a vertex or bend is linearly interpolated between its position in $P$ and in $Q$, so the point travels on a straight line at constant speed.  A \defn{unidirectional morph} is a special case of a linear morph where the lines along which the points move are all parallel. 
In a \defn{horizontal [vertical] morph}, these lines are horizontal [vertical].

Our algorithm works with poly-line drawings, and the algorithm may add \defn{degenerate bends} at existing vertices/bends %
or at a grid point along an edge segment.  The algorithm may also delete such degenerate  bends.  %
We separate the steps of a linear morph sequence into: steps that add/delete degenerate bends; and steps that linearly interpolate between drawings that have the same number of bends along every edge, in which case their orders of appearance along the edge in the two drawings determines their correspondence.

In an \defn{orthogonal point drawing} %
vertices are represented by points and edges by orthogonal poly-lines. 
In an \defn{orthogonal box drawing} a vertex is represented by a positive area \defn{box} with two horizontal sides and two vertical sides, and edges are represented by orthogonal poly-lines.
The point where an edge attaches to a vertex box is called a \defn{port}. 
An orthogonal box drawing is \defn{planar} if there are no extraneous feature intersections: no coincident ports, no intersecting vertex boxes,  no intersections between an edge and another edge, no intersection between an edge and a vertex box except at a port. 
One exception is that we will allow a port at a corner of a vertex box.
For orthogonal 
box drawings, there is one more type of \defn{degenerate bend} in addition to those defined above: a bend at a port. %

We use the following relationship between box drawings and poly-line drawings.  Any 
planar orthogonal box drawing has a 
corresponding 
planar
\defn{admitted poly-line drawing}, where the vertex box is replaced by a vertex point at the center of the box joined by straight segments to the ports, which become bends in the edges.  Thus an edge drawn with $k$ bends in the orthogonal box drawing becomes an edge with $k+2$ bends in the admitted drawing. 
\Cref{fig:overview-alg2} shows extra segments for the admitted drawings with dotted lines.

A morph of an orthogonal box drawing $D$ is specified in terms of 
its \defn{defining points} which are the vertex box corners, the ports, and the edge bends. 
In a \defn{linear morph} between two orthogonal box drawings the positions of the defining points are linearly interpolated. 
We restrict to \defn{structure-preserving} linear morphs that meet the following conditions:
\begin{enumerate}
\item At every time during the morph the positions of the four corners of a vertex box determine a \defn{vertex rectangle} which is a rectangle of positive area, though not necessarily with horizontal and vertical sides, and ports stay attached to their vertex rectangle.
\item 
For an edge from $u$ to $v$, every segment along the edge remains orthogonal (horizontal or vertical) and does not change its direction (upward, downward, rightward or leftward).  
If a segment has length zero at a strictly intermediate time during the morph, then it has length zero throughout the morph.
\end{enumerate}
By allowing non-axis-aligned vertex rectangles in condition (1) we can 
``twist'' a square vertex box as shown in
\cref{fig:simple-twist-example}\footnote{This is similar to the ``rotation'' steps used to morph rectangular duals~\cite{chaplick2023morphing}.}.
The figure shows a clockwise twist where each corner moves to the position of its clockwise neighbouring corner, so two of the corners move horizontally and the other two move vertically.  Ports and bends may move non-orthogonally.
Twists will be defined formally in Section~\ref{sec:performing-twists}. 
Apart from twists, all our morphs
of orthogonal box drawings
will in fact be
horizontal or vertical.

\begin{figure}[ht]
\hspace*{\fill}
\includegraphics[page=6,scale=0.9]{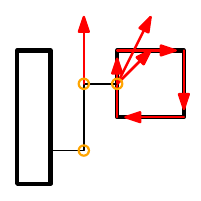}
\hspace*{\fill}
\includegraphics[page=7,scale=0.9]{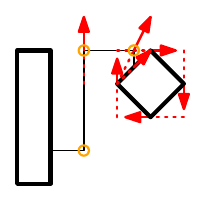}
\hspace*{\fill}
\includegraphics[page=8,scale=0.9]{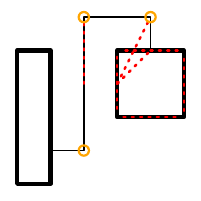}
\hspace*{\fill}
\caption{An example of a twist at a vertex box. 
Red arrows show motion to come; dashed portions show past motion.}
\label{fig:simple-twist-example}
\end{figure}

A linear morph of an orthogonal box drawing is \defn{planarity-preserving} if the drawing is planar at every time during the morph.
When we say a \defn{planarity-preserving linear morph of an orthogonal box drawing} we always mean that the morph is also structure-preserving.
We prove in Section~\ref{sec:reduction} that a planarity-preserving linear morph of an orthogonal box drawing induces a planarity-preserving linear morph of its admitted poly-line drawing. 

We  morph orthogonal box drawings using a sequence of intermediate goals.
Two orthogonal box drawings of the same graph  are \defn{port-aligned} if,
\changed{for each vertex $v$,}
when we compare the two vertex boxes representing 
$v$, then  for every side (top, bottom, right left) the sequence of edges attached to that side of the vertex box is the same in both drawings. 
Two drawings
without degenerate bends
are \defn{parallel} if they are port aligned and for every directed edge 
the two directed poly-lines representing the edge have the same sequence of left/right turns.
To make drawings parallel, we use twists and zig-zag eliminations (to be defined), where 
a \defn{zig-zag} is a sequence of two opposite turns, left then right, or right then left.

\section{Algorithm Overview}
\label{sec:overview}

This section gives the main steps of the algorithm that will prove Theorem~\ref{thm:main}.
The input is a pair of compatible straight-line planar drawings of a graph $G$ on $n$ vertices.

\begin{samepage}
\begin{algorithm-stmt}[Reduction to orthogonal box drawings]
\label{alg:reduction}
\end{algorithm-stmt}
Morph each of the straight-line drawings to an admitted poly-line drawing of an  orthogonal box drawing
in which edges are drawn as single horizontal/vertical segments. 
See \cref{fig:overview-alg1} for an example.
We prove in Lemma~\ref{lemma:orthogonal-box-linear-morphs-to-admitted-linear-morphs} that a morph of an orthogonal box drawing induces a morph of its admitted poly-line drawing, so this step reduces the problem
to that of morphing orthogonal box drawings
with zero bends per edge.
Details of the reduction are in Section~\ref{sec:reduction}.
\end{samepage}

\begin{algorithm-stmt}[Morphing between planar orthogonal box drawings]
\label{morph-algorithm}
\end{algorithm-stmt}
The input consists of a pair of compatible planar orthogonal box drawings $P$ and $Q$ of a graph $G$ on $n$ vertices.
$P$ and $Q$ may have bends, which is more general than the output of \cref{alg:reduction}. 
We describe steps of the morph ``from both ends'', modifying $P$ and $Q$ until they meet in the middle, which is legitimate since linear morphs are reversible.

\begin{description}
\item{\bf Phase I:}
Morph $P$ and $Q$
to become parallel.
Building on the approach of Biedl et al.~\cite{biedl2013morphing} for the case of orthogonal point drawings, we use the following steps:
\begin{description}
\item{\bf Phase Ia:} Morph $P$ to become port-aligned with $Q$.
This adds bends in edges of~$P$.
\item{\bf Phase Ib:} Morph $P$ and $Q$ to eliminate zig-zags.
(In the situation %
required to prove Theorem~\ref{thm:main}, $Q$ has no bends and thus no zig-zags.)
\item{\bf Phase Ic:} Morph $P$ so that for every directed edge  the difference between its number of left and right turns (the ``spirality'') is the same in $P$ as in $Q$,
while preserving port-alignment. 
This is accomplished via twists, which are  
interspersed  with zig-zag eliminations.
We use the property  that if $P$ and $Q$ are port-aligned with matching spirality and no zig-zags or degenerate bends, then they are parallel.
\end{description}
Details of Phase I can be found  in Section~\ref{sec:morph-orth-box}.

\item{\bf Phase II}: Morph between the two parallel orthogonal box drawings.  To do this, we imagine ports, corners, and bends as vertices and  
appeal directly to the  result of Biedl et al.~\cite{biedl2013morphing} for morphing parallel orthogonal point drawings.  
See Appendix~\ref{appendix:Phase-II} for  details.
\end{description}

See \cref{fig:overview-alg2} for a sketch of all these steps.
Although most of the detailed proofs that our morphs are planarity-preserving appear only in the Appendix, it is worth noting that, apart from twists (which are handled in Section~\ref{sec:performing-twists})
and the reduction to orthogonal box drawings, 
all of our morphs are horizontal or vertical, which simplifies the task.
We prove   
in Lemma~\ref{lemma:horizontal-morph} in Appendix~\ref{appendix:preserving-planarity} that 
a horizontal morph between planar orthogonal box drawings $P$ and $Q$ is planarity-preserving so long as every horizontal line intersects the same sequence of defining points and edge/box segments.

\section{Reduction to Orthogonal Box Drawings}
\label{sec:reduction}

In this section we 
morph a straight-line planar graph drawing to an admitted poly-line drawing of an orthogonal box drawing,  and we prove that a morph between two orthogonal box drawings induces a morph between their admitted drawings.  
In all cases ``morph'' means satisfying all our conditions, as detailed below.
We begin with the second result:

\begin{lemma}
\label{lemma:orthogonal-box-linear-morphs-to-admitted-linear-morphs}
Let $P$ and $Q$ be
planar orthogonal box drawings of the same graph.
Let $P'$ and $Q'$ be the admitted planar poly-line drawings of $P$ and $Q$, respectively.
Suppose 
there is a planarity-preserving linear morph 
from $P$ to $Q$.
Then the linear morph from $P'$ to $Q'$ is also planarity-preserving.
\end{lemma}

The idea is that a linear combination is preserved throughout a linear morph, and the center of a vertex box/rectangle is a linear combination (the average) of the four corners.  For details, see Appendix~\ref{appendix:reduction}.

We now turn to the problem of  morphing a straight-line planar drawing to an admitted drawing of an orthogonal box drawing.

\begin{theorem}
\label{thm:planar-straight-line-to-boxes-morph}
Let $P$ be a planar straight-line drawing of a graph $G$ with $n$ vertices drawn on an $O(n) \times O(n)$ integer grid.
Then there is
a planar orthogonal box drawing $D$ of $G$ on an $O(n) \times O(n)$ grid with no bends, and there is
a planarity-preserving linear morph sequence
of length $O(n)$
from $P$ to $P'$, the admitted poly-line drawing of $D$,  where each explicit intermediate drawing lies on an $O(n) \times O(n)$ grid.  
Moreover, 
this sequence can be computed in $O(n^2)$ time.  
\end{theorem}

\begin{figure}[ht]
\hspace*{\fill}
\begin{minipage}{0.3\linewidth}
\includegraphics[page=23,scale=0.9]{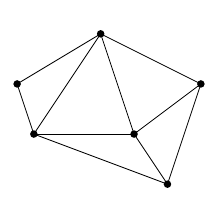}

\vspace*{-10mm} (a)
\end{minipage}
\hspace*{\fill}
\begin{minipage}{0.35\linewidth}
\includegraphics[page=24,scale=0.9]{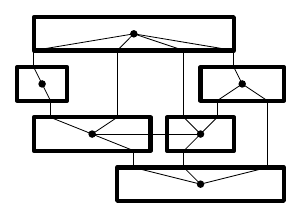} 

\vspace*{-10mm} (b)
\end{minipage}
\hspace*{\fill}
\begin{minipage}{0.3\linewidth}
\includegraphics[page=25,scale=0.9]{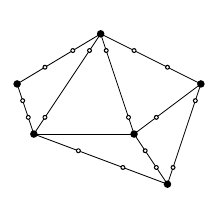} 

\vspace*{-10mm} (c)
\end{minipage}
\hspace*{\fill}
\caption{(a) a straight-line drawing $P$;
(b)
the corresponding orthogonal box drawing $D$ and admitted drawing $P'$ with every vertex at the same $y$-coordinate as in $P$; (c) adding two bends to each edge of $P$ would allow a horizontal morph to $P'$, but this is forbidden since the bends are not on the grid.   
}
\label{fig:straight-line-to-admitted-and-box-example}
\end{figure}

We prove Theorem~\ref{thm:planar-straight-line-to-boxes-morph} 
in two steps: (1) construct an orthogonal box drawing $D$ with certain properties as described below; 
(2) show how to morph from $P$ to $P'$.

For (1) a vertex $v$ of $P$ at $y$-coordinate $c$ will be replaced in %
$D$ by a vertex box that is skinny in the $y$-direction: 
we add a new grid line just below $c$ and a new grid line just above $c$ and construct a vertex box 
between these new grid lines.  Adding the grid lines \defn{refines} the grid in the $y$ direction by a factor of $3$. 
See \cref{fig:straight-line-to-admitted-and-box-example}.
The properties for $D$ are that edges have no bends, each vertex has the same $y$-coordinate in $P$ and in the refined grid of $P'$, and every horizontal line intersects the same sequence of edges in $P$ and $P'$. 
We construct $D$ in $O(n\log n)$ 
time by turning an existence result due to Biedl~\cite{biedlheight}
into an efficient algorithm.  Her 
result says that 
any planar straight-line drawing can be redrawn so that each vertex becomes a horizontal segment (or ``bar'') \changed{with the} %
same $y$-coordinate, and each non-horizontal edge becomes a vertical segment incident to the bars representing its endpoints, while maintaining the order of edges and vertices crossed by any horizontal line.  
Thickening the vertex bars to boxes gives the result we need. 
For details, see Appendix~\ref{appendix:reduction}.

For (2), 
we show how to morph from $P$ to the admitted drawing $P'$. 
It would be straight-forward to add degenerate bends where each edge crosses the newly added horizontal grid lines (see Figure~\ref{fig:straight-line-to-admitted-and-box-example}(c)), and perform a single horizontal morph to move those bends to the position of the corresponding ports.
However, our rule is that new bends can only be added at grid points, so we are forced to a more complicated solution.  Imagine $P'$ drawn strictly to the left of $P$.  
Order the edges and vertices of $P$ left-to-right, with priority given to edges. 
One by one, in this order, pull the edges and vertices of $P$ to their positions in $P'$. To handle one edge of $P$ we add two new bends at its endpoints, and morph them to the positions of the corresponding ports.
See \cref{fig:straight-line-to-admitted-and-box-algorithm-example}
for a %
sketch of the algorithm and
Appendix~\ref{appendix:reduction} for more details.

\begin{figure}[ht]
\hspace*{\fill}
\includegraphics[page=6,scale=0.45]{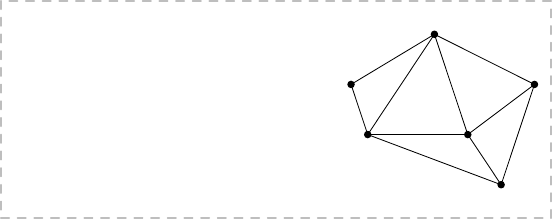}
\hspace*{\fill}
\includegraphics[page=7,scale=0.45]{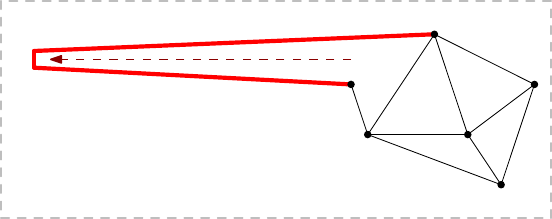}
\hspace*{\fill}
\includegraphics[page=8,scale=0.45]{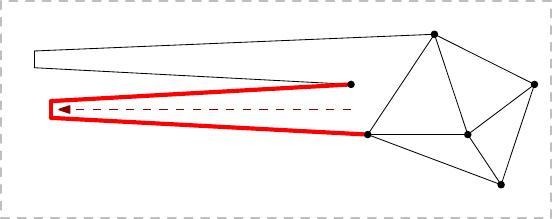}
\hspace*{\fill} \\
\vspace{1em}%
\hspace*{\fill}
\includegraphics[page=9,scale=0.45]{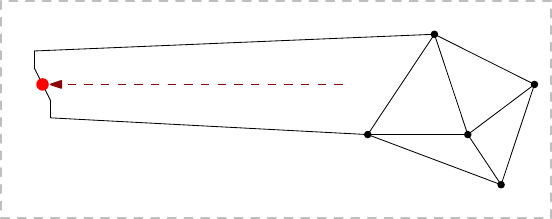}
\hspace*{\fill}
\includegraphics[page=28,scale=0.45]{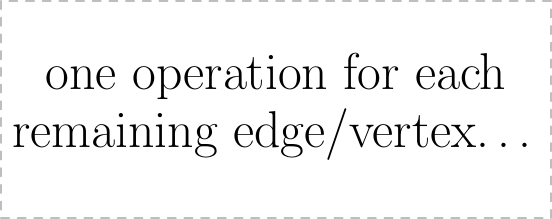}
\hspace*{\fill}
\includegraphics[page=21,scale=0.45]{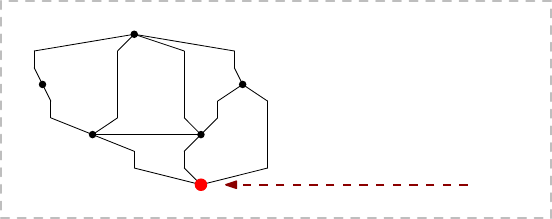}
\hspace*{\fill}
\caption{
\changed{Morphing from a straight-line drawing on the right to
(the admitted drawing of)
an orthogonal box drawing on the left by moving edges/vertices one-by-one.
Each pane shows the newly morphed edge/vertex in red.}
}
\label{fig:straight-line-to-admitted-and-box-algorithm-example}
\end{figure}

\section{Morphing Orthogonal Box Drawings}
\label{sec:morph-orth-box}

This section contains
the algorithm to morph between orthogonal box drawings 
(Algorithm~\ref{morph-algorithm}).  
The algorithm works for any two compatible planar orthogonal box drawings, but the bounds needed to prove our main theorem (number of linear morphs, run-time, grid size, and number of bends) apply only when the input drawings lie on an $O(n) \times O(n)$ grid and have $O(1)$ bends per edge.
Informally, we call these \defn{small} drawings but lemma statements are precise.
Let \defn{$f(D)$} denote the number of defining points of orthogonal box drawing $D$. A count of box corners, ports, and bends shows that 
$f(D)$ is $O(n)$ if the number of bends is $O(n)$.

Recall from Section~\ref{sec:overview} that the algorithm has two phases. Phase I morphs the drawings $P$ and $Q$ to become parallel, and is described in this section.  
Phase II morphs between the resulting parallel drawings, and is deferred to  Appendix~\ref{appendix:Phase-II},  
since it follows directly from Biedl et al.~\cite{biedl2013morphing}.

Phase Ia, port alignment, is in 
Section~\ref{sec:Phase-Ia}.
Port aligned drawings are parallel if each directed edge has the same sequence of left/right turns in both drawings.
Following Biedl et al.~\cite{biedl2013morphing}, we define the
\defn{spirality} of a directed orthogonal poly-line to be its number of left turns minus its number of right turns\footnote{This is different from the definition of spirality in~\cite{van2022optimal}.} ignoring degenerate bends.
An edge drawn without zig-zags 
has only left turns or only right turns.  Thus if the two drawings of an edge have the same spirality, no degenerate bends, and no zig-zags, then they have the same sequence of left/right turns.
We eliminate zig-zags in Section~\ref{sec:zig-zag-elimination} and adjust spirality in Section~\ref{sec:spirality}. 
A main novel contribution is the twist operation for boxes in Section~\ref{sec:performing-twists}.  

\subsection{Phase Ia: Port Alignment}
\label{sec:Phase-Ia}

\begin{figure}[ht]
    \centering
\hspace*{\fill}
    \includegraphics[scale=0.3,page=10]{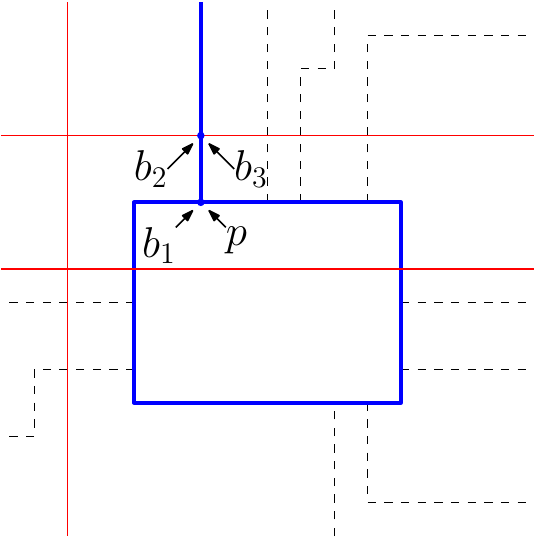}
\hspace*{\fill}
    \includegraphics[scale=0.3,page=11]{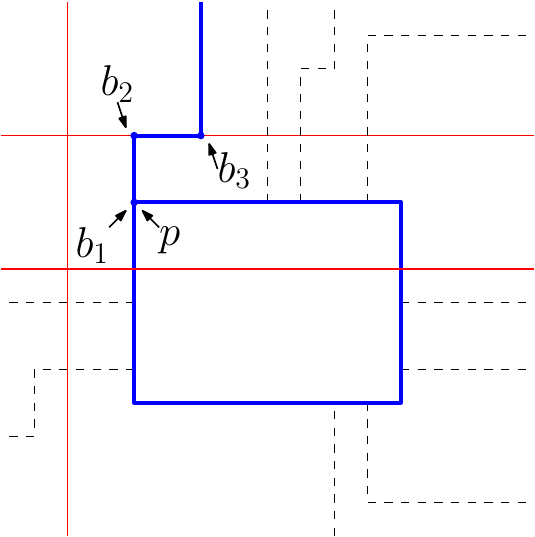}
\hspace*{\fill}
    \includegraphics[scale=0.3,page=12]{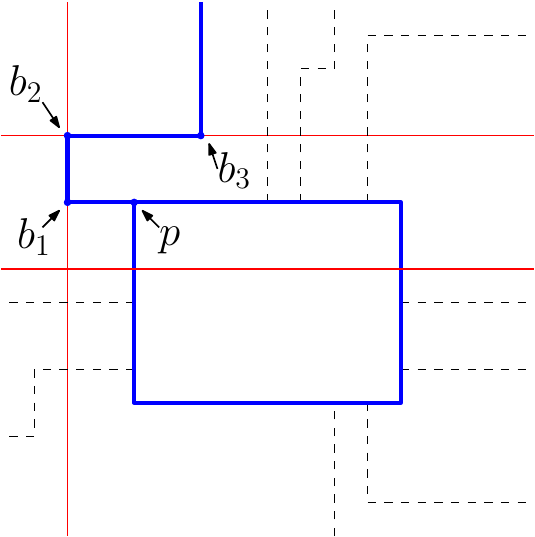}
\hspace*{\fill}
    \includegraphics[scale=0.3,page=13]{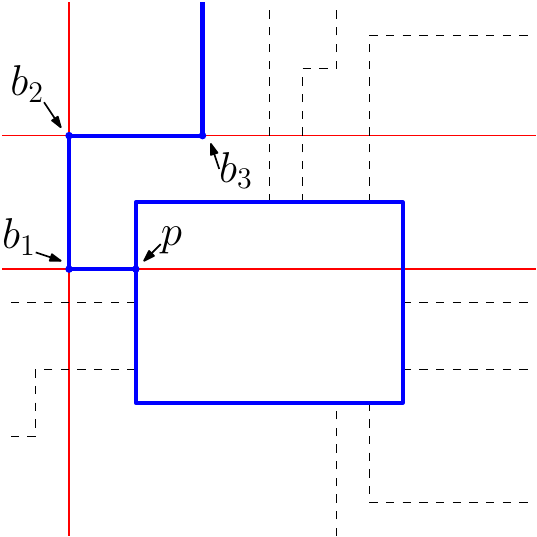}
\hspace*{\fill}
    \caption{A sketch of how a port can be morphed around a corner by introducing 3 bends $b_i$ and 3 new grid lines (in red).}
    \label{fig:corner-morph-steps}
\end{figure}

We morph $P$ to be port-aligned with $Q$ by successively ``turning'' a port around the corner of a vertex box as shown in
\cref{fig:corner-morph-steps}.  Each turn operation uses $O(1)$ planarity-preserving linear morphs, and adds 3 bends and 3 grid lines. 
We prove in Appendix~\ref{appendix:Phase-Ia}
that $O(1)$ such turn operations per edge suffice.
For small drawings, we use in total $O(n)$ linear morphs, and remain on an $O(n) \times O(n)$ grid with $O(1)$ bends per edge. The run-time is $O(n^2)$.
Note that a port may be coincident with a corner of a vertex box in explicit intermediate drawings. However, after Phase Ia, we never allow 
a port at a corner.

\subsection{Phase Ib: Zig-zag Elimination}
\label{sec:Phase-Ib}
\label{sec:zig-zag-elimination}
Zig-zag elimination
is used in Phases Ib and Ic.
It 
performs a sequence of horizontal and vertical morphs to
get rid of all the zig-zags in an orthogonal box drawing, and produces a drawing $D$ with  
grid size $f(D)$. 
This  is
$O(n)$ if each edge has $O(1)$ bends, %
but the exact bound 
$f(D)$
will 
be %
important 
when we apply zig-zag elimination $O(n)$ times in Phase Ic.

In Phase Ib, we apply zig-zag elimination once to $P$, and, if necessary,  once to $Q$.  
Note that during Phase I, the drawing $Q$ is altered only in Phase Ib, and is not altered at all 
in the situation arising from Theorem~\ref{thm:main}.

A zig-zag is \defn{horizontal} [\defn{vertical}] if the segment between the two turns is horizontal [vertical].
Biedl et al.~\cite{biedl2013morphing}
were the first to show that a %
horizontal linear morph can be
used to 
eliminate one 
horizontal 
zig-zag in an orthogonal point drawing. 
Figure~\ref{fig:zigzag}
shows the idea as applied to an orthogonal box drawing.
Van Goethem et al.~\cite{van2022optimal}
observed
that all horizontal zig-zags in a drawing $P$ 
can be eliminated simultaneously  with a single 
horizontal linear morph to a drawing $P'$.    
Their result also applies to orthogonal box drawings, but we must improve the grid size and runtime, see below.

\begin{figure}[ht]
\hspace*{\fill}
\includegraphics[page=1,scale=0.6]{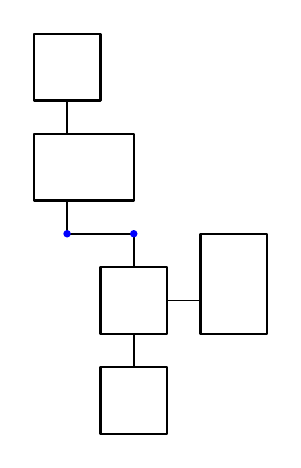}
\hspace*{\fill}
\includegraphics[page=3,scale=0.6]{figure-zig-zag-elimination-example}
\hspace*{\fill}
\includegraphics[page=4,scale=0.6]{figure-zig-zag-elimination-example}
\hspace*{\fill}
\caption{An orthogonal box drawing with a horizontal zig-zag, and
the zig-zag elimination that converts its segment into two coinciding
degenerate bends.}
\label{fig:zigzag}
\end{figure}

One could determine $P'$ by pretending to perform $O(n)$ individual zig-zag eliminations in $P$, but this would be too slow 
and result in a large grid.
Instead we compute $P'$
via \emph{compaction
techniques}.
We %
encode the relative $x$-order of all maximal vertical segments of a drawing $P$ in an auxiliary directed acyclic graph $G_T(P)$ 
of size at most $f(P)$. 
Then we modify $G_T(P)$ 
to capture
that the segments of horizontal zig-zags must vanish, and determine the target-drawing $P'$ simply by computing longest
paths in $G_T(P)$ to get the $x$-coordinates of defining 
points ($y$-coordinates remain unchanged).     
Graph $G_T(P)$ can be computed in $O(n\log n)$ time using a so-called
trapezoidal map; the run-time can be reduced to $O(n)$ when $P$
represents a connected graph using Chazelle's  
triangulation results
\cite{chazelle1991triangulating}.   
The computed $x$-coordinates are at most 
the number of defining points of $P'$, not counting degenerate bends (which we are about to eliminate).
More details 
of this morph to eliminate all horizontal 
zig-zags
are given in 
Lemma~\ref{lemma:horizontal-zigzag-elimination} 
in Appendix~\ref{appendix:Phase-Ib}.

By 
repeatedly eliminating horizontal zig-zags and then eliminating vertical zig-zags
we can eliminate all zig-zags, which yields 
the main result of this section:

\begin{lemma}
\label{lemma:zig-zag-elimination}
Let $G$ be a planar graph with $n$ vertices.  Let $P$ be a planar orthogonal box drawing of $G$ on an  $O(n) \times O(n)$ grid with $O(1)$  bends per edge.
Then there is a planarity-preserving linear morph sequence of length $O(1)$ from $P$ to a zig-zag-free orthogonal box drawing $P'$ with the same port alignment and edge spiralities as $P$ such that 
$P'$ is drawn on an $f(P') \times f(P')$ (hence $O(n) \times O(n)$) grid.
Furthermore, 
during the linear morph sequence, 
the number of bends never increases, 
and the width and height of the grid may increase but never beyond $f(P)$.
Finally, 
the linear morph sequence can be computed in $O(n)$ time.
\end{lemma}
\begin{proof}
Repeat the following four steps (a ``round'') until no bends have been removed for an entire round:
\begin{enumerate}
\item eliminate all horizontal zig-zags with one horizontal morph as described above,
\item remove degenerate bends,
\item eliminate all vertical zig-zags with one vertical morph,
\item remove degenerate bends.
\end{enumerate}

\begin{figure}
\hspace*{\fill}
    \begin{subfigure}[t]{0.45\textwidth}
        \centering
        \includegraphics[page=1,scale=0.6]{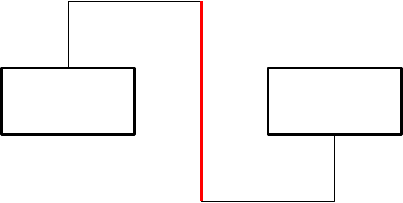}
        \caption{A drawing $P$ with a vertical zig-zag and no horizontal zig-zags.}
    \end{subfigure}
\hspace*{\fill}
    \begin{subfigure}[t]{0.45\textwidth}
        \centering
        \includegraphics[page=2,scale=0.6]{figure-zig-zag-elimination-creates-zig-zag-example}
        \caption{A drawing $P'$ with a horizontal zig-zag resulting from a zig-zag elimination on $P$.}
    \end{subfigure}
\hspace*{\fill}
    \caption{An example of a vertical zig-zag elimination that results a new horizontal zig-zag.}
    \label{fig:zig-zag-elimination-creates-zig-zag-example}
\end{figure}

Note that we need 
multiple rounds
in general because eliminating vertical zig-zags may create horizontal
\changed{ones, and vice versa (see \cref{fig:zig-zag-elimination-creates-zig-zag-example}).}
Except for the last round, 
each round eliminates at least one zig-zag---and thus at least two bends---from every edge that has zig-zags.  No bends are ever added.  Thus the number 
of rounds (hence also the number of morphs) is $O(1)$.   Each drawing in the linear morph sequence can be computed in $O(n)$ time, so the run-time to compute all of them is also in $O(n)$.

Each morph   preserves port alignment and edge spiralities.   
Steps 1 and 2 do not change the grid height and result in a drawing of width at most the (current) number of defining points, which is at most the initial number $f(P)$.
Symmetrically, steps 3 and 4 do not change the grid width and result in a drawing of height at most the current number of defining points, which is at most $f(P)$.
At the end, the final drawing $P'$ lies on an $f(P') \times f(P')$ grid since the last round did not decrease the number of defining points.
\end{proof}

\subsection{Phase Ic: Adjusting Spirality}
\label{sec:Phase-Ic}
\label{sec:spirality}

In Phase Ic we morph $P$ so that the spirality of each edge matches its spirality in $Q$.
The basic operation we use is a ``twist'', a morph that essentially turns all the ports
around a vertex box: in a clockwise twist the ports on the top of the box become ports on the right, while ports on the right become ports on the bottom, and so on around the box. 
(Note that doing this 
by turning edges around corners as in Section~\ref{sec:Phase-Ia} would require a linear number of morphs.)
We describe twists in Section~\ref{sec:performing-twists}.

How do twists adjust spirality? We will see %
that a clockwise twist of a vertex box adds two left turns and one right turn to all edges leaving the vertex box. 
This guides the calculation of how many twists to apply at each vertex box.  
We also need a multiple of four twists for each vertex box
in order to maintain port alignment. 
The question of how many twists to apply and what order to apply them in is addressed in Section~\ref{sec:planning-twists}. 
We also keep the drawing zig-zag free, and on a small grid.

\subsubsection{Performing twists}
\label{sec:performing-twists}

In this section we focus on how to perform one twist simultaneously on a set of vertex boxes. 
We need a preliminary set-up step that makes all the vertex boxes square and clears some space around them.  More precisely, the \defn{$k$-proximal region} of the vertex box of vertex $v$ 
consists of all points within $L_\infty$ distance $k \cdot \deg(v)$ from the box, where $\deg(v)$ is the degree of $v$ (thus the $k$-proximal region is a rectangle formed by enlarged the vertex box by $k \cdot \deg(v)$ on each side). 
The vertex boxes of a drawing are \defn{$k$-spaced} if their $k$-proximal regions are disjoint and the only edge segments intersecting a $k$-proximal region are those incident to the vertex box.
With a constant number of linear morphs we can 
make all vertex boxes square and $2$-spaced 
(as shown in the first pane of \cref{fig:twist-rotation-construction}), while keeping the drawing port-aligned, and preserving the spirality of each edge.  
See Lemma~\ref{lemma:square-spaced} in Appendix~\ref{appendix:performing-twists}.

\begin{figure}
    \centering
    \begin{subfigure}[t]{0.49\textwidth}
        \centering
\includegraphics[page=1,scale=0.24]{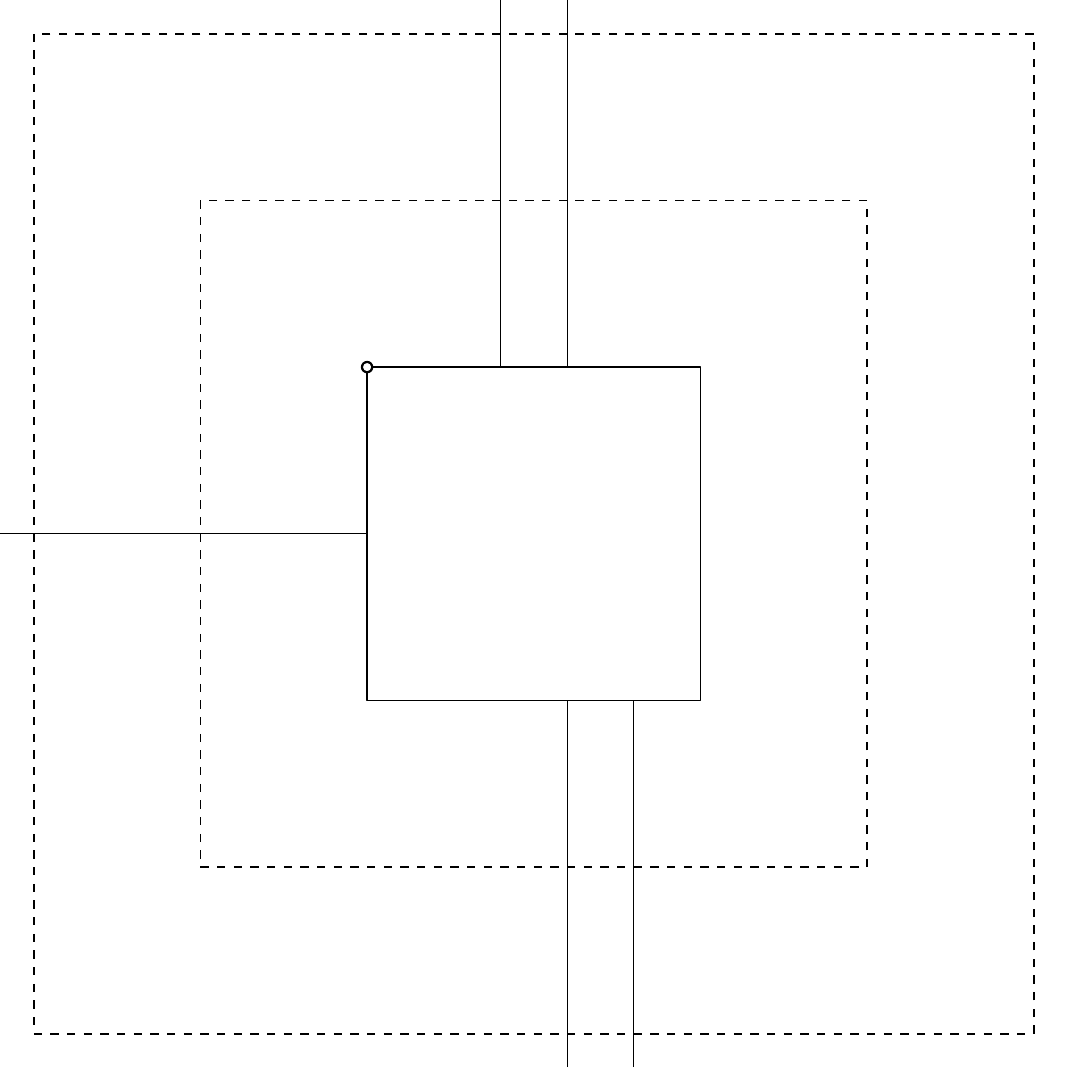}
    \end{subfigure}
    \begin{subfigure}[t]{0.49\textwidth}
        \centering
\includegraphics[page=2,scale=0.24]{figure-twist-construction}
    \end{subfigure}\\
    \vspace{2em}
    \begin{subfigure}[t]{0.49\textwidth}
        \centering
\includegraphics[page=3,scale=0.24]{figure-twist-construction}
    \end{subfigure}
    \begin{subfigure}[t]{0.49\textwidth}
        \centering
\includegraphics[page=4,scale=0.24]{figure-twist-construction}
    \end{subfigure}
    \caption{
    Performing a clockwise twist.
    Bends are drawn as small red 
    circles.
    The top left box corner (marked with a hollow circle) moves to the top right.
}
    \label{fig:twist-rotation-construction}
\end{figure}

We next describe how to take an orthogonal box drawing that has square 2-spaced vertex boxes, and  perform one twist of one vertex box. 
The twist will only alter the drawing inside the $2$-proximal region of the vertex box, so we can twist multiple vertices simultaneously, see Lemma~\ref{lemma:simultaneous-twist} below.

A  \defn{clockwise twist} operation is defined as follows, see \cref{fig:twist-rotation-construction}.
(Counterclockwise twists are defined analogously.)
Let $v$ be a vertex of degree $\deg(v)$ and let $S$ be its square vertex box.
In the first step of the morph we add 3 degenerate bends in each edge incident to $S$.  The second step is the actual linear morph.
We give exact coordinates for the bends added to the $i^{\rm th}$ edge $e_i$ (in clockwise order)
that is attached to the top side of $S$, say at port $p_i$.  Other coordinates are analogous.
Bend $b_{i,1}$ is placed at $p_i$.  Bends $b_{i,2}$ and $b_{i,3}$ are placed on the edge $e_i$, at distance $\deg(v) + i$ above~$p_i$.

During the linear morph, each corner of $S$ moves to the next clockwise corner.  Each port on $S$ is a linear combination of two box corners;  
this is maintained
throughout the morph. 
It remains to describe the positions of the bends.
As before, we give exact coordinates only for edge $e_i$ \changed{attached to the top side}.  
Let $p'_i$ be the new position for port $p_i$.  
Bend $b_{i,3}$ does not move, so $b_{i,3}'=b_{i,3}$.  Bend $b_{i,1}$ moves to a point $b'_{i,1}$ that is $i$ units to the right of $p'_i$. 
Bend $b_{i,2}$ moves horizontally to a point $b'_{i,2}$ directly above $b'_{i,1}$. 

\begin{lemma}
\label{lemma:planar-twist}
The clockwise twist %
is structure-preserving and planarity-preserving. 
\end{lemma}

\begin{proof}[Proof outline.]
We need only focus on the 2-proximal region of $S$.  
We first show \changed{that} the twist is structure-preserving.  
Because $S$ is square, it remains square throughout the morph.  Also, each port stays attached to its side of $S$.  The edge segment $p_i b_{i,1}$ remains horizontal, starting at length zero and ending at length $i$.  Segment $b_{i,1}b_{i,2}$ remains vertical of positive length, and segment $b_{i,2} b_{i,3}$ remains horizontal, starting at length zero and ending at positive length.

We next show \changed{that} the twist is planarity-preserving. 
It follows from the above that edges do not intersect the vertex box during the morph.  
It remains to show that no two edges intersect. 
For edges incident to the same side (say the top side) 
of $S$, observe that:  the vertical segments $b_{i,1}b_{i,2}$ maintain their left-to-right ordering; the horizontal segments $b_{i,2}b_{i,3}$ remain at constant, distinct $y$-coordinates, and the horizontal segments $p_i b_{i,1}$ maintain their top-to-bottom ordering.
For edges incident to different sides of $S$, we define four disjoint subregions of the 2-proximal region of $S$ as shown in
\cref{fig:twist-rotation-construction-region},
and argue that edges remain in their subregion as we morph  
the edges and the subregions (details \changed{are} 
in Appendix~\ref{appendix:performing-twists}).
\end{proof}

\begin{figure}
    \centering
    \begin{subfigure}[t]{0.32\textwidth}
        \centering
        \includegraphics[page=17,scale=0.2]{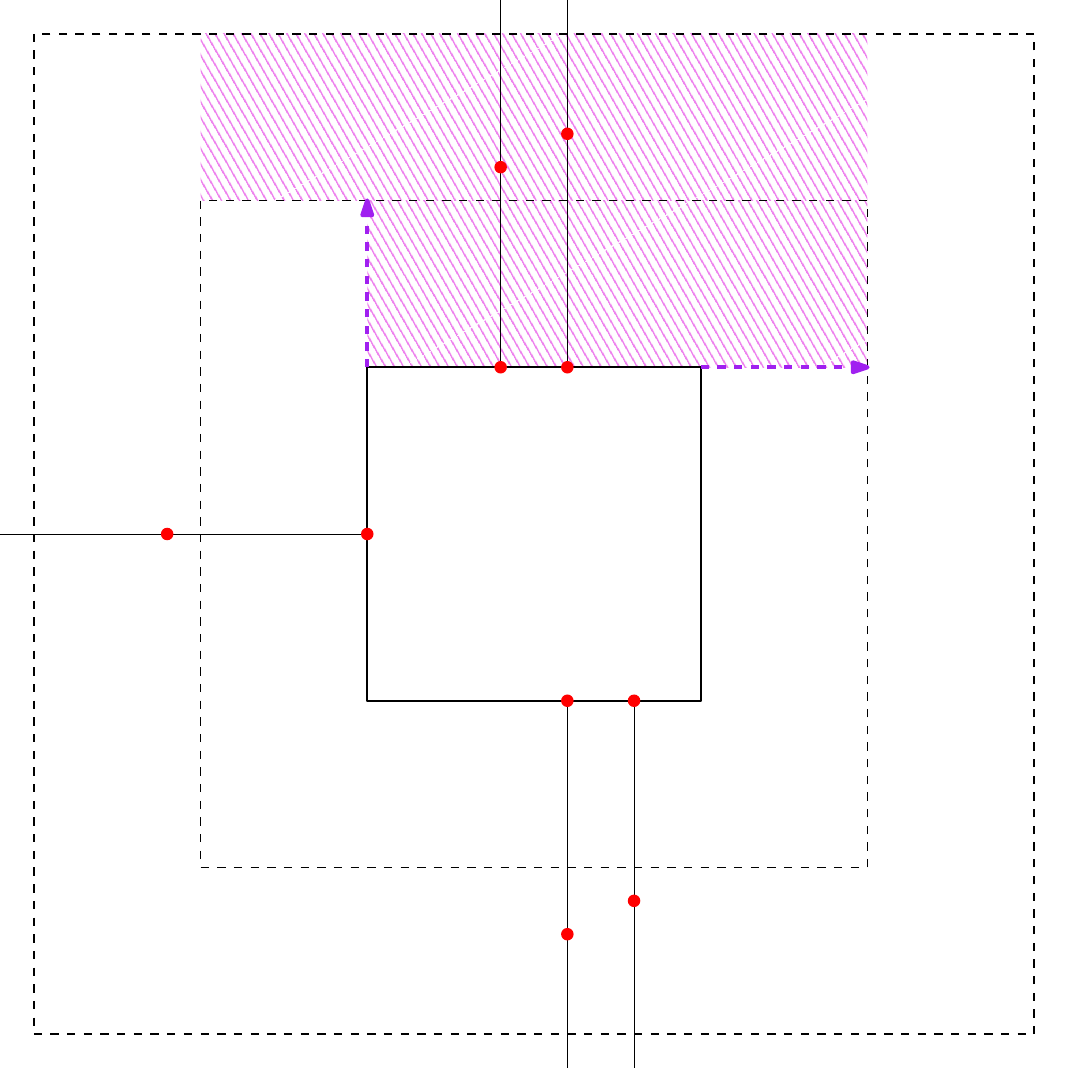}
    \end{subfigure}
    \hfill
    \begin{subfigure}[t]{0.32\textwidth}
        \centering
        \includegraphics[page=18,scale=0.2]{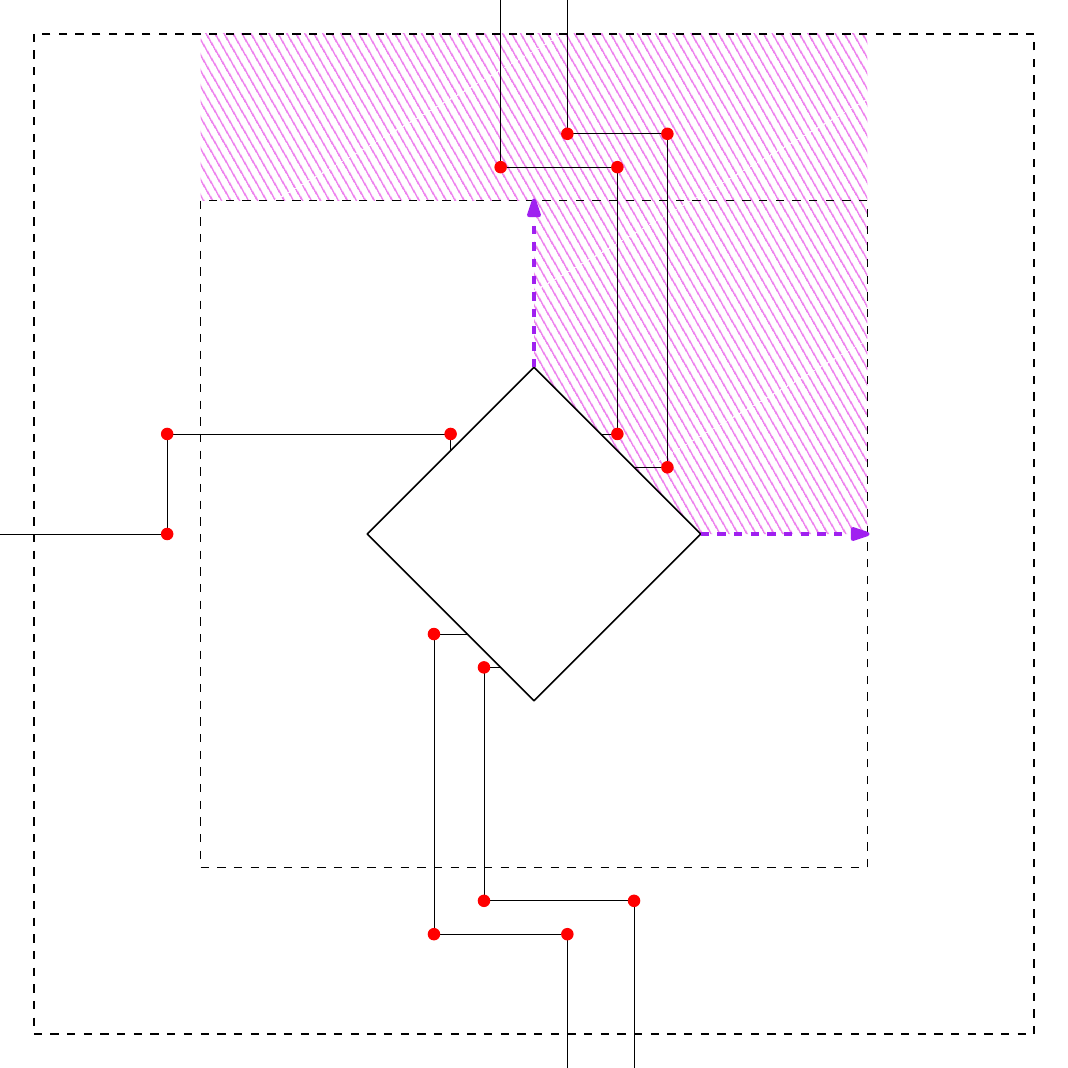}
    \end{subfigure}
    \hfill
    \begin{subfigure}[t]{0.32\textwidth}
        \centering
        \includegraphics[page=19,scale=0.2]{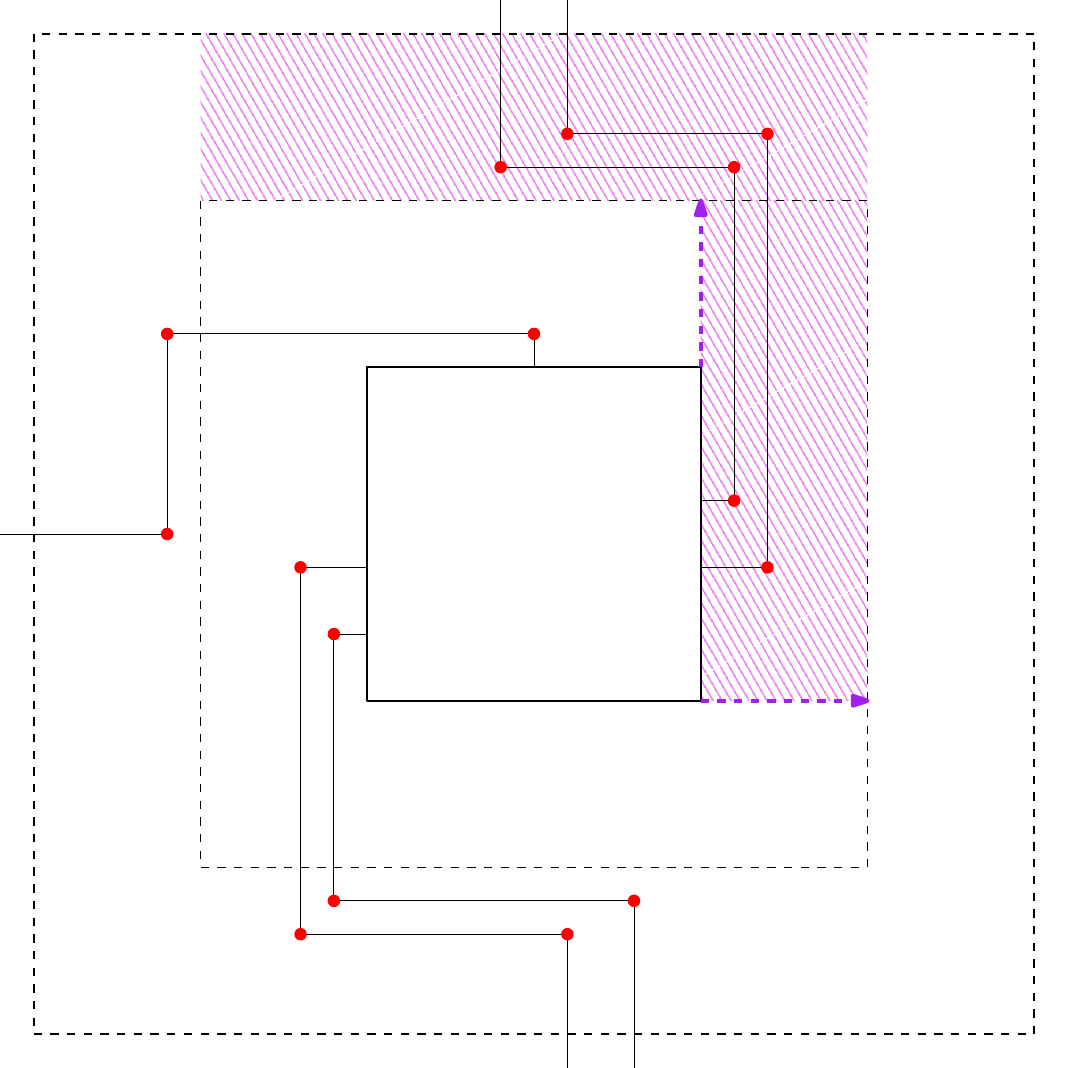}
    \end{subfigure}
    \caption{The subregion of the $2$-proximal region of $S$ corresponding to the top side of an initial box, throughout the morph.}
    \label{fig:twist-rotation-construction-region}
\end{figure}

Because a twist does not change the drawing outside the 2-proximal region of a vertex box, we can indeed (as promised before) twist multiple vertices simultaneously,
which immediately gives the following:

\begin{lemma}%
\label{lemma:simultaneous-twist}
Let $G$ be a connected planar graph with $n$ vertices.
Let $P$ be a planar orthogonal box drawing of $G$ drawn on an 
$O(n) \times O(n)$ grid 
with O(1) bends per edge. 
Let $t : V (G) \rightarrow \{-1, 0, 1\}$ specify the twists to be performed at vertices, where $-1$ indicates a clockwise twist and $1$ indicates a counterclockwise twist.
Then there is a planarity-preserving sequence of length $O(1)$ from $P$ to an orthogonal box drawing $P'$ that effects the twists of $t$, 
where 
each explicit intermediate drawing is drawn on an 
$O(n) \times O(n)$ grid with $O(1)$ bends per edge.
Moreover, 
the linear morph sequence can be computed in $O(n^2)$ time.
\end{lemma}

\subsubsection{Planning twists}
\label{sec:planning-twists}

There are two aspects to planning twists in $P$: (1) decide how many twists to perform at each vertex box; and (2) schedule the twists into rounds where each round twists each vertex box at most once, so that one round can be performed simultaneously as in Lemma~\ref{lemma:simultaneous-twist}. 
For (1) we solve a set of equations, the same as in Biedl et al.~\cite{biedl2013morphing}. Step (2) is new, as are the improved bounds on \changed{the}
number of linear morphs and \changed{the} grid size that complete Phase Ic.

\subparagraph{Numbers of twists.} We need some notation. For edge $e=uv$ and any orthogonal box drawing $D$, let $s_D(u,v)$ denote the spirality of the edge in $D$.
Our goal is to morph $P$ (while leaving $Q$ unchanged) until for every edge $uv$ its \defn{spirality difference}, $\Delta s_P(u,v) := s_P(u,v) - s_Q(u,v)$, becomes zero.
The final $P$ must also be port-aligned with $Q$ (equivalently, port-aligned with the initial $P$), %
although this need not be true of the explicit intermediate drawings.
These two conditions together constitute the requirements for our twists.

Observe that a clockwise twist 
(as defined in the previous section) at $u$'s vertex box adds two left turns and one right turn to every edge leaving $u$, for a net change of $+1$ in the spirality of edges $uv$ leaving $u$, and a net change of $-1$ for edges $vu$ entering $u$. These are reversed for a counterclockwise twist. 
For example, if $\Delta s_P(u,v)$ is positive, we should apply counterclockwise twists to $u$ and/or clockwise twists to $v$.
For vertex $v$, let $t(v)$ be the number of twists to be performed on $v$'s vertex box in $P$, where a positive value indicates counterclockwise twists and a negative value indicates clockwise twists. 
Necessary and sufficient conditions to meet the requirements for our twists are:
\begin{alignat}{2}
&\text{for each edge $uv$} &&\quad \Delta s_{P}(u,v) = t(u) - t(v)
\label{eqn:edge-twist-condition}\\
&\text{for each vertex $v$} && \quad t(v) \equiv 0 \bmod 4
\label{eqn:vertex-twist-condition}
\end{alignat}

\begin{lemma} A solution to the above equations with $|t(v)| \in O(n)$ can be found in $O(n)$ time.
\end{lemma}
\begin{proof}[Proof outline]
The lemma was proved by 
Biedl et al.~\cite{biedl2013morphing} for orthogonal point drawings. 
As a main ingredient,
they prove that it suffices to impose condition (\ref{eqn:vertex-twist-condition}) for a single vertex $v_0$, and to impose condition (\ref{eqn:edge-twist-condition}) for the edges of a spanning tree $T$ rooted at $v_0$.  The idea is that a non-tree edge $e$ creates a cycle with $T$ which corresponds to a face of the drawing. Condition (\ref{eqn:edge-twist-condition}) is then proved for $e$ using the fact that 
a planar orthogonal cycle has 4 more right turns than left turns when traversed clockwise.
We prove that this main ingredient carries over to orthogonal box drawings---the difference is that a face of the point drawing has edges meeting at vertices, whereas a face of a box drawing contains portions of box sides. See Lemma~\ref{lemma:twist-equations} in Appendix~\ref{appendix:planning-twists}. 

Given this 
main ingredient,
we assign $t(v_0) := 0$, and 
assign $t(v)$ to be $t(v_0)$ plus the sum of the spirality differences on edges of the path from $v_0$ to $v$ in $T$. 
This takes $O(n)$ time, and gives $t(v) \in O(n)$ since the spirality difference of each edge is constant.
\end{proof}

\subparagraph{Allocating twists to rounds.} 
We allocate the twists $t(v)$ to rounds $r_i$, where $i$ ranges from $1$ to 
$k
:= \max_v |t(v)|$ 
so that each vertex box is twisted at most once in each round. 
For vertex $v$, we put its twists into rounds 1 through $|t(v)|$.  

\subparagraph{Algorithm for Phase Ic.} Let $P'_0 := P$.  For $i =1, \ldots, 
k$,
perform the simultaneous twists of round $r_i$ on input $P'_{i-1}$ according to Lemma~\ref{lemma:simultaneous-twist} to obtain $P_i$, and then do zig-zag elimination on $P_i$ according to Lemma~\ref{lemma:zig-zag-elimination} to obtain $P'_i$.
Output $P'_k$.

\subparagraph{Analyzing the algorithm.}
Note that the algorithm does $O(n)$ linear morphs.
It remains to bound the  grid size,  the number of bends, and the run-time for Phase Ic.  Since zig-zag elimination is called $O(n)$ times it becomes important to give absolute bounds  to avoid hiding a ``growing constant'' inside the  $O(n)$ bound on grid size.
The main idea is to analyze the number of bends, since that determines the grid size after zig-zag elimination.

\begin{lemma}
There are constants $c$ and $d$, independent of $i$, such that $P_i$ and $P'_i$, $i=1,\ldots, k$ lie on a grid of size $c n \times c n$ with at most $d$ bends per edge.     
\end{lemma}
\begin{proof}
We first prove 
a bound on bends for $P'_i$ for $i=0,\dots,k$.  
The drawing $P'_i$ is zig-zag-free so the number of bends on edge $uv$ is the absolute value of its spirality, which is the same as in $P_i$ since zig-zag elimination preserves spirality  
(Lemma~\ref{lemma:zig-zag-elimination}).  
Our method of allocating twists to rounds ensures that the spirality of an edge stays the same except during a constant number of rounds where it improves.  
More formally: 

\begin{claim}[Proof in Appendix~\ref{appendix:planning-twists}]
\label{claim:spirality-decrease}
After each round of simultaneous twists, the  spirality of an 
edge remains the same or gets closer to its spirality in $Q$,
i.e., $s_{P_i}(u,v)$ lies in the closed interval between 
$s_{P_{i-1}}(u,v)$ and $s_Q(u,v)$.
\end{claim}

Thus $|s_{P_i}(u,v)| \le \max \{|s_P(u,v)|, |s_Q(u,v)|\}$
which is bounded by some constant $d'$ %
since edges in $P$ and $Q$ each have a constant number of bends.

We now turn to the grid size.
For $i=0,\dots,k$
drawing $P'_i$ is the result of applying
Lemma~\ref{lemma:zig-zag-elimination} and hence 
its grid size is $f(P'_i) \times f(P'_i)$, where $f(P'_i)$ is the number of defining points of $P'_i$.
Counting box corners, ports and bends, gives $f(P'_i) \le 4n + 2m + d' m$, where $G$ has $m \le 3n-6$ edges, so $f(P'_i)\leq c'n$ for some
$c'\leq 10+3d'$.
$P_i$ is obtained from $P'_{i-1}$ by applying one round of simultaneous twists as in  
Lemma~\ref{lemma:simultaneous-twist}. 
This may add a constant term to the bound $d'$ on the number of bends, and may increase the bound $c'n\times c'n$ on the grid size by a constant factor,
which gives the final values of $d$ and $c$. %
\end{proof}

We summarize Phase Ic as follows:

\begin{lemma}
\label{lemma:match-spirality}
Let $G$ be a connected planar graph with $n$ vertices. Let $P$ and $Q$ be compatible port-aligned zig-zag-free planar orthogonal box drawings of $G$ on an 
$O(n) \times O(n)$ grid with $O(1)$ bends per edge.  
Then there is a planarity-preserving linear morph sequence of length $O(n)$ 
from $P$ to a zig-zag-free orthogonal box drawing $P'$ with the same port-alignment and edge spiralities as $Q$ (thus $P'$ and $Q$ are parallel) such that all explicit intermediate drawings are on an
$O(n) \times O(n)$ grid with $O(1)$ bends per edge. 
Moreover, the linear morph sequence can be computed in $O(n^2)$ time.
\end{lemma}

\changed{\section{Conclusions}}

We have now assembled all 
the ingredients
for proving our two main theorems. 
In summary,
Theorem~\ref{thm:box-morph} (morphing planar orthogonal box drawings) is proved by the steps of Algorithm~\ref{morph-algorithm}: 
Phase Ia (port-alignment) is sketched in Section~\ref{sec:Phase-Ia}, with details in \cref{lemma:port-alignment} in Appendix~\ref{appendix:Phase-Ia};
Phase Ib (zig-zag elimination in both drawings) is %
handled by Lemma~\ref{lemma:zig-zag-elimination};
Phase Ic (spirality) is 
handled by
Lemma~\ref{lemma:match-spirality};
and Phase II (morphing parallel drawings) simply appeals to \cite{biedl2013morphing} (details are in \cref{lemma:phase-ii} in \Cref{appendix:Phase-II}).

Theorem~\ref{thm:main} (morphing straight-line drawings on a small grid) is proved by the reduction to morphing orthogonal box drawings (see Theorem~\ref{thm:planar-straight-line-to-boxes-morph}
in Section~\ref{sec:reduction}) plus Theorem~\ref{thm:box-morph}.

In conclusion, our main result is an algorithm
to morph compatible planar straight-line drawings with a linear number of linear morphs for which all explicit intermediate drawings lie on small grids, have few bends per edge, and can be computed quickly.
Extending the result to disconnected graphs involves further details, and is left for future work.

The biggest remaining open problem is the one from the introduction: Is there a piece-wise linear morph with small explicit intermediate
drawings \emph{without} bends?
If so, can 
we limit to a short linear morph sequence?
And if not, can we at least limit the number of bends per edge to a very small constant, such as 1 or 2?

\bibliography{arxiv.bib}

\appendix

\section{Details on Preserving Planarity}
\label{appendix:preserving-planarity}
Alamdari et al.~\cite{alamdari2017morph}
and Kleist et al.~\cite[Lemma 7]{kleist2019convexity}
proved that
a unidirectional morph of a straight-line drawing preserves planarity if it preserves the ordering of edges/vertices crossed by any line parallel to the morph direction. 
We extend this (for horizontal and vertical morphs) to poly-line drawings and to orthogonal box drawings, allowing degenerate bends and ports at corners.

The original result easily
extends to these other drawings styles if there are no coincident defining points
(i.e., no degenerate bends, and no ports at box corners)
at any time during the morph, because we can simply 
make all the bends, corners, and ports
into vertices.  This gives two straight-line drawings to which the original result applies.  
However, if there are coincident defining points at any time during the morph, then this transformation creates coincident vertices, which violate planarity.
And note that two defining points may be coincident at some time during the morph, and not at other times, so we cannot replace two coincident defining points by a single vertex. 

Consider two poly-line or orthogonal box drawings $P$ and $P'$ with the same number of bends on each edge (thus a correspondence between the bends according to their order along each edge). 
For a horizontal line $\ell$ directed left to right we say that the orderings of $P$ and $P'$ along $\ell$ are \defn{compatible} if 
no two line segments (edge or box segments) or defining points 
$a$ and $b$ appear along $\ell$ 
with $a$ strictly before $b$
in one drawing and 
$b$ strictly before $a$
in the other.   

\begin{lemma}
\label{lemma:horizontal-morph}
Suppose that $P'$ is obtained from $P$ via a horizontal morph, and that the orderings of $P$ and $P'$ are compatible for every horizontal line.  Then the morph from $P$ to $P'$ is planarity-preserving, and also structure-preserving in the case of orthogonal box drawings.    
\end{lemma}
\begin{proof}
We begin by showing that the morph is structure-preserving for orthogonal box drawings. If some points are vertically aligned in both $P$ and $P'$, then they are vertically aligned throughout the morph.
Therefore boxes remain boxes, and ports along left/right sides stay attached to their vertex boxes.
A port on the bottom side of a vertex box lies in the closed interval between the two bottom corners in both $P$ and $P'$, so by compatible orderings, the port stays attached to the bottom throughout the morph.
A similar argument applies to ports on the top side of a vertex box. 
The next condition for structure-preserving is about length $0$ segments.
If a segment (of an edge or a box) has length $0$ at an intermediate time $t \in (0,1)$ of the morph, then we claim that it 
has length $0$ throughout the morph, 
because otherwise its two endpoints $p$ and $q$ are horizontally aligned and switch order during the morph, a contradiction to the compatible orderings.
This also implies that if two defining points are coincident at any intermediate time during the morph, then they are coincident throughout the morph, a fact that will be used below.
Finally, segments 
cannot change their directions (upward, downward, leftward, rightward).

Next, we show that the morph preserves planarity.
Observe that, since $P$ is planar, we can choose $t_1 >0$ small enough that no violation of planarity 
has yet occurred during the morph from time 0 to time $t_1$. 
(The space of planar representations is open.)
Symmetrically, since $P'$ is planar, we can choose $t_2$, with $t_1 < t_2 <1$, large enough that no violation of planarity occurs between times $t_2$ and $1$.

Let $P_i$ be the drawing at time $t_i$, $i=1,2$. 
It remains to show that the morph from $P_1$ to $P_2$ preserves planarity. 
As noted above, if two defining points are coincident  at any time during 
the morph from $P_1$ to $P_2$, 
then they are coincident throughout this whole morph.
We can then do a transformation to straight-line drawings 
similar to the one described above.  In 
$P_1$ and $P_2$, 
replace every bend, box corner, and port by a vertex, and then repeatedly
remove coincident vertices 
via contractions.
Call the resulting straight-line graphs $P'_1$ and $P'_2$, respectively.  
Observe  that $P'_1$ and $P'_2$ are compatible planar drawings and any horizontal line crosses the same ordering of edges and vertices in both.
Thus we can 
apply the result of  Kleist et al.~to $P'_1$ and $P'_2$  to prove that planarity is preserved throughout the morph from $P_1$ to $P_2$, and hence throughout the entire morph from $P$ to $P'$.
\end{proof}

\section{Details for Section~\ref{sec:reduction}: Reduction to Orthogonal Box Drawings}
\label{appendix:reduction}

We first record a basic property of linear morphs:
\begin{observation} \ 
\label{obs:linear-combinations}
A linear morph preserves linear combinations: 
If a point $q$ is the same linear combination of  points 
$p_1, \ldots, p_k$ at the beginning and the end of a linear morph, then $q$ is the same linear combination of those points at every time $t \in [0,1]$ of the morph.

\end{observation}

\begin{proof}[Proof of Lemma~\ref{lemma:orthogonal-box-linear-morphs-to-admitted-linear-morphs}]
The admitted drawings $P'$ and $Q'$ differ from the orthogonal box drawings $P$ and
$Q$ (respectively) only within the vertex boxes, so it suffices to show that during the morph
from $P'$ to $Q'$
the point representing a vertex
remains in the strict interior of its vertex box.  At the beginning and the end of the morph the vertex is located at the center of its vertex box, i.e., at the average of the four corners. By Observation~\ref{obs:linear-combinations}, the vertex is located at the center of the vertex box throughout the morph.
\end{proof}

\begin{proof}[Proof of Theorem~\ref{thm:planar-straight-line-to-boxes-morph}]
We separate the proof into: (1) constructing $D$, and (2) morphing from $P$ to $P'$.   

\subparagraph{Part (1)}
Recall that we can obtain $D$ easily
if we have the drawing $D^{\text{bar}}$ whose existence was proved 
in 
\cite{biedlheight}: Each vertex is a horizontal bar (with the same $y$-coordinate
as in $P$), each edge is a horizontal or vertical line segment, and every horizontal line intersects the vertices and edges in the same
left-to-right order in $P$ and $D^{\text{bar}}$.   %
Such a drawing is also called a \emph{(weak 2-directional flat) visibility representation}.

To find $D^{\text{bar}}\!$, we first add a \emph{frame}, which is a 6-cycle that surrounds the bounding box $B$ of $P$, with a vertex $s$ 
placed above the topmost row of $B$, vertex $t$ placed below the bottommost row of $B$, and four vertices placed adjacent to the corners of $B$ (one unit to the left at the left corners, and one unit to the right at the right corners).
See also Figure~\ref{fig:findVR}.
Triangulate all inner faces of the resulting
straight-line drawing; this can be done in
$O(n\log n)$ time~\cite{mark2008computational}.
Next, delete all horizontal edges and call the result $P^+$.
We will convert $P^+$ into a visibility representation $D^+$ with the same $y$-coordinates and left-to-right orders;
this will then give $D^{\text{bar}}$ by deleting the added vertices and edges
and re-inserting horizontal edges, which is feasible since the left-to-right orders are preserved.

\def\scalefactor{1}
\begin{figure}[ht]
\hspace*{\fill}
\includegraphics[scale=\scalefactor,page=2,trim=25 0 31 0,clip]{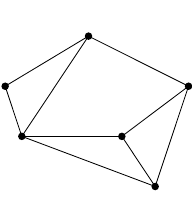}
\hspace*{\fill}
\includegraphics[scale=\scalefactor,page=4,trim=18 0 23 0,clip]{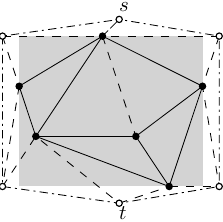}
\hspace*{\fill}
\vspace{0.5em}
\\
\hspace*{\fill}
\includegraphics[scale=\scalefactor,page=6]{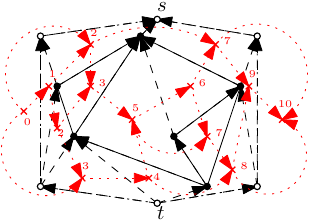}
\hspace*{\fill}
\includegraphics[scale=\scalefactor,page=7]{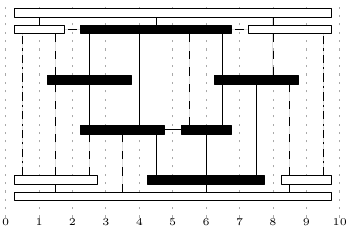}
\hspace*{\fill}
\caption{The original drawing; adding a frame (dash-dotted) around the bounding box (gray) and triangulating (dashed); the modified dual graph of $P^+$ and lengths of longest paths from the source (which has length $0$);
and the resulting visibility representation. }
\label{fig:findVR}
\end{figure}

To find $D^+$, first direct all edges in $P^+$ upward; since we triangulated the drawing (before deleting horizontal edges) every vertex except $s$ and $t$ then has incoming
and outgoing edges.    Thus $P^+$ is an \emph{$st$-graph}: It is a directed acyclic planar graph with a unique source and unique sink that are
both on the outer-face.   It is well-known how to create visibility representations of such graphs \cite{rosenstiehl1986rectilinear,tamassia1986unified}:  The $y$-coordinates
can be any integer-assignment that is compatible with the edge-directions, i.e., the head of a directed edge has larger $y$-coordinate than the tail. The $x$-coordinates can be obtained in linear time by taking the dual graph  and directing its edges left-to-right according to the edge-directions in $P^+$.
Then split the outer-face vertex into a source and a sink, which makes the dual graph acyclic as well.
Finally compute longest paths from the source in the resulting acyclic graph; the $y$-coordinates of vertex corners and edge-segments can then be directly read from the coordinates of incident faces.
Therefore a drawing $D^+$ can be computed in linear time, by using the same $y$-coordinates as in $P$ (which are compatible with the edge-directions since we directed upward).   The construction is known to preserve the planar embedding and therefore the left-to-right orders.
This proves~(1).

\subparagraph{Part (2)}
We show how to morph from $P$ to $P'$.
Refer to Figure~\ref{fig:straight-line-to-admitted-and-box-algorithm-example} for an outline and to Figure~\ref{fig:straight-line-to-admitted-and-box-algorithm-example-more} for more
steps
of the procedure.
First translate $P'$ horizontally so it is strictly to the left of $P$.
We will morph vertices and edges of $P$ one-by-one to their locations in the translated $P'$. 
Order the non-horizontal edges and vertices 
of $P$ so that: $a<b$ if some horizontal line directed left-to-right %
crosses edge/vertex $a$ strictly 
before it crosses edge/vertex $b$; and $e<v$ if edge $e$ is incident to vertex $v$.  These constraints are acyclic, so they can be extended to a total order, and we morph the elements in this order.

\begin{figure}[ht]
\hspace*{\fill}
\includegraphics[page=10,scale=0.45]{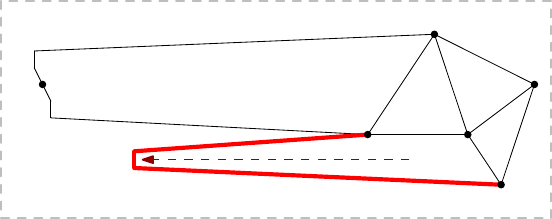}
\hspace*{\fill}
\includegraphics[page=27,scale=0.45]{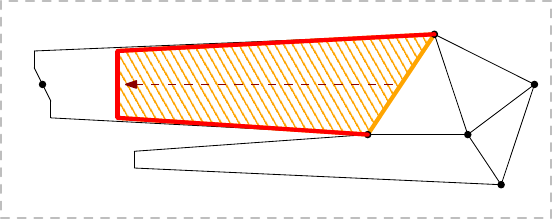}
\hspace*{\fill}
\includegraphics[page=12,scale=0.45]{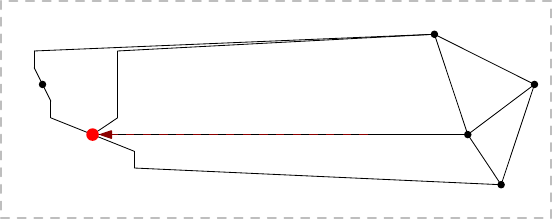}
\hspace*{\fill}
\caption{
More steps of morphing from a straight-line drawing to the admitted drawing of an orthogonal box drawing.   The striped region indicates the convex quadrilateral swept when morphing the edge.
}
\label{fig:straight-line-to-admitted-and-box-algorithm-example-more}
\end{figure}

When it is an edge's turn to morph, we add two degenerate bends at the endpoints of the edge in $P$, and morph these bends to their positions in $P'$.  Note that the endpoints of $e$ remain in their positions in $P$. 
When it is a vertex's turn, we morph the vertex from its position in $P$ to its position in $P'$.  This is a horizontal morph since the vertex lies at the same $y$-coordinate in the two drawings. 
Note that horizontal edges do not appear in the ordering, but
a horizontal edge acquires no bends and its position is just determined by the positions of its endpoints. 

We prove by induction that each step of this morph 
preserves planarity 
and preserves the ordering of edges and vertices along any horizontal line. 
This is true when a vertex moves, since those morphs are horizontal, so we can apply Lemma~\ref{lemma:horizontal-morph}.
Next, consider what happens when we morph a non-horizontal edge $e=v_1 v_2$, introducing bends $b_i$ at $v_i$, $i=1,2$
that do not move horizontally.  
Suppose $v_1$ is below $v_2$ in $y$-coordinate. 
Let $P_e$ be the drawing just after $e$ is morphed together with the line segment representing $e$ in $P$ (i.e., an overlay of ``before'' and ``after'' morphing $e$). 
Consider a horizontal line $\ell$ between $v_1$ and $v_2$, directed left-to-right. 
In $P_e$, the line $\ell$ first may cross edges and vertices that come before $e$ in the total order in their new positions in $P'$; then $e$, as drawn in $P'$; then $e$, as drawn in $P$; then edges and vertices that come after $e$ in the total order in their original positions in $P$.  In particular, the convex quadrilateral formed by the two drawings of $e$ in $P$ and in $P'$ is empty of other parts of the drawing $P_e$. As we morph $e$, its drawing remains inside this quadrilateral.
Therefore morphing $e$ preserves planarity and preserves the ordering of edges and vertices along any horizontal line. 
\end{proof}

\section{Details for Section~\ref{sec:Phase-Ia}: Port Alignment}
\label{appendix:Phase-Ia}

We show here that the same port alignment can be achieved with few morphs while adding a constant number of bends to each edge.

\begin{lemma}
\label{lemma:port-alignment}
Let $G$ be a planar graph with $n$ vertices.
Let $P$ and $Q$ be compatible orthogonal box drawings of $G$ on an 
$O(n) \times O(n)$ grid with $O(1)$ bends per edge. 
Then there is a planarity-preserving linear morph sequence of length $O(n)$ from $P$ to an orthogonal box drawing $P'$ that is port-aligned with $Q$, where each explicit intermediate drawing is on an $O(n) \times O(n)$ grid and has $O(1)$
bends per edge. %
Moreover, 
the linear morph sequence can be found in $O(n^2)$ time.   
\end{lemma}

\begin{proof}
Figure~\ref{fig:corner-morph-steps} gave one key ingredient:   With a  constant number of horizontal and vertical morphs we can ``turn'' one port around one corner of a vertex box by adding 3 bends and 3 grid lines. 
We claim that $O(n)$ operations of turning a port around a box-corner suffice to make $P$ port-aligned with $Q$.  Consider vertex box $B$ with $d$ ports. One simple to state (but sub-optimal) method is as follows. 
Let $p$ be the port that should be leftmost on the top side of $B$ (pick a different side in case the top side has no ports in $Q$).  We can get $p$ in the right place (pushing other ports ahead of it) with 
at most $2d$ clockwise or counterclockwise turn operations in total and at most 2 per port. 
After that, each port needs to do at most 3 clockwise or counterclockwise turns to get to the correct place. 
\end{proof}

\section{Details for Section~\ref{sec:Phase-Ib}: Zig-zag Elimination}
\label{appendix:Phase-Ib}
We show here how to eliminate all horizontal zig-zags with a single linear morph that can be computed efficiently and where the resulting drawing is small.

\begin{lemma}
\label{lemma:horizontal-zigzag-elimination}
Let $G$ be a connected planar graph with $n$ vertices.
Let $P$ be an orthogonal box drawing of $G$ drawn on an $O(n) \times O(n)$ grid with $O(1)$ bends per edge.
There exists a (single) planarity-preserving horizontal morph from $P$ to an orthogonal box drawing $P'$ 
with no horizontal zig-zags and
with the same port-alignment and edge spiralities as $P$.
Moreover, 
$P'$ has no more bends than $P$, 
it has the same grid height as $P$,  
its grid width is at most the number $f$ of defining points of $P'$ not counting degenerate bends,  
and it can be computed in $O(n)$ time.
\end{lemma}
\begin{proof}
Let $L$ be the maximal vertical segments of $P$ (meaning that we combine vertical segments that share an endpoint); these are disjoint segments and    $|L|\leq \tfrac{1}{2} f(P)\in O(n)$ since every end of a segment in $L$ is a defining point of $P$.
The \defn{trapezoidal map} of $L$ is created by extending horizontal line segments/rays to the left and right of each endpoint 
of a segment in $L$ %
until they hit another line segment of $L$ (see~\cite{mark2008computational}).
The \defn{trapezoidal graph} $G_T(P)$ has a vertex for each line segment of $L$ and an edge, directed left to right, for each finite horizontal line segment of the trapezoidal map.
The trapezoidal graph can be computed in $O(n \log n)$ time with a simple line sweep, and
Chazelle's linear-time triangulation algorithm~\cite{chazelle1991triangulating}
improves this to $O(n)$ for the line segments of a simple polygon.
In our situation, we compute the trapezoidal graph for each face of $P$ separately
with Chazelle's algorithm and combine them to obtain $G_T(P)$ in linear time.
Note that faces of $P$ are polygons without holes by connectivity. Some faces may only be weakly simple (for example around a vertex box of a degree 1 vertex), but we can perturb them to simple polygons.
We can also deal with the unbounded outer face by adding a frame (similar to what was done in Figure~\ref{fig:findVR}). 

The trapezoidal graph $G_T(P)$ encodes all required constraints on $x$-coordinates of segments of $L$ to ensure planarity:   for any assignment of new $x$-coordinates to segments in $L$ we obtain a planar drawing as long as the head of any directed edge has a larger $x$-coordinate than the tail.    This is the insight at the heart of algorithms for \emph{compacting VLSI designs} introduced by Doenhardt and Lengauer \cite{doenhardt1987algorithmic}.    They compute for any orthogonal box-drawing a constraint graph $C_{\leq}$ (not exactly the same as $G_T(P)$, but the two graphs have the same transitive closure and hence encode the same set of constraints).    Then they compute (in linear time) the length of the longest path from a source
\changed{to} 
each vertex of the constraint graph.    This gives one possible set of feasible $x$-coordinates, and all values of the coordinates are between $0$ and $|L|-1$ (and often much smaller).

We can use the same approach, with a small modification, to compute our desired drawing $P'$ in linear time by using $G_T(P)$ as the constraint graph and modifying it suitably.
Consider one horizontal zig-zag consisting of a horizontal segment from bend $b_1$ to $b_2$ with incident maximal vertical segments $\ell_1$ and $\ell_2$, respectively.  A horizontal morph to eliminate the zig-zag joins $\ell_1$ and $\ell_2$ into a single maximal vertical segment.
In the constraint graph, this means that vertices $\ell_1$ and $\ell_2$ are merged, and
edges generated exclusively by $b_1$ and $b_2$ are removed.
See Figure~\ref{fig:zig-zag-elimination-trapezoidal-map-effect-example}.

\renewcommand{\floatpagefraction}{.8}%

\begin{figure}
    \centering
\includegraphics[page=1,scale=0.7]{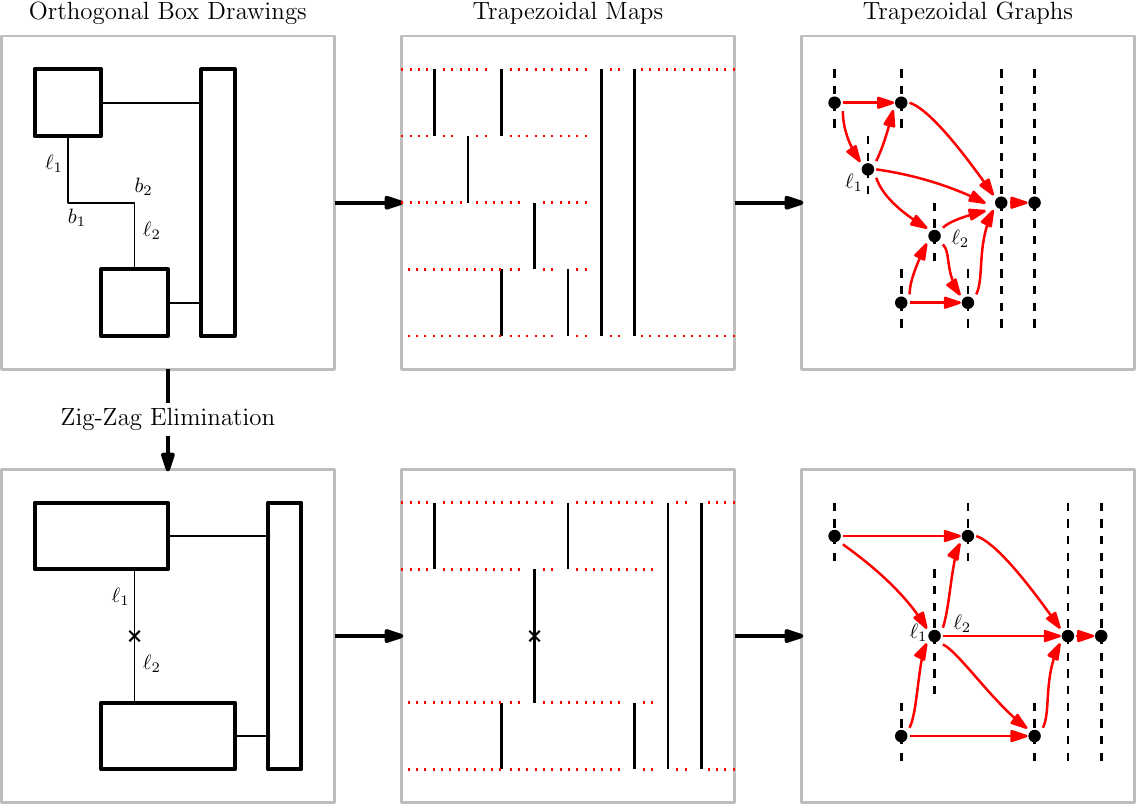}
    \caption{An example of how eliminating one horizontal zig-zag
    \changed{affects}
    the trapezoidal graph.  Two vertical segments become one segment, with degenerate bends marked by a cross.  (Our algorithm does not actually perform this zig-zag elimination, but modifies the trapezoidal graph as if it had been performed.) %
    }
    \label{fig:zig-zag-elimination-trapezoidal-map-effect-example}
\end{figure}

Thus we can obtain the trapezoidal graph $G_T(P')$ of the
as-yet-unknown
drawing $P'$ by contracting, for every horizontal zig-zag, the two vertices of the participating vertical segments and deleting some of the incident edges.   This can be done in overall linear time by performing the contractions in increasing order of the $y$-coordinates of the segments of the horizontal zig-zags and always merging the vertex of the upper vertical segment into the vertex of the lower vertical segment; in this way edges change endpoints at most twice (once per endpoint) during all contractions.

We know that the resulting set of constraints is satisfiable, because we could pretend to execute each horizontal zig-zag elimination separately to obtain some drawing that satisfies all of them.   Our drawing $P'$ is then obtained by using the same $y$-coordinates as in $P$, and using as $x$-coordinates the length of the longest paths from a source in $G_T(P')$.

We verify that $P'$ satisfies all properties.   Since $y$-coordinates are unchanged, 
every horizontal line $\ell$ intersects the same segments in $P$ and $P'$, and the constraint set guarantees that the orderings of $P$ and $P'$ along $\ell$ are compatible.   Therefore by 
Lemma~\ref{lemma:horizontal-morph} the linear morph from $P$ to $P'$ (which is a horizontal morph) is planarity-preserving.  
The modification from $G_T(P)$ to $G_T(P')$ forces the two bends of a horizontal zig-zag to coincide, so all horizontal zig-zags are eliminated.
We have no ports at corners when performing zig-zag eliminations, so edges remain adjacent to the same (and unique) side of their incident vertex boxes throughout, 
which together with the compatible orderings implies that $P$ and $P'$ have the same port assignments.
Eliminating a zig-zag from an edge removes one left turn and one right turn, so the spirality of the edge is preserved.

The claims on the bends and grid size are also easily verified: we never add bends, and $P$ and $P'$ have the same height since $y$-coordinates are unchanged. 
The grid width of drawing $P'$ is at most $|G_T(P')|=|L|-z$ (where $z$ is the number of horizontal zig-zags).   Also $|L|\leq \tfrac{1}{2}f(P)$ and $f=f(P)-2z$, so the width is at most $\tfrac{1}{2}f$ as desired.
\end{proof}

\section{Details for Section~\ref{sec:performing-twists}: Performing Twists}
\label{appendix:performing-twists}

We now describe how to modify an orthogonal box drawing to have enough space to make a twist operation feasible.

\begin{lemma}
\label{lemma:square-spaced}
Let $P$ be an orthogonal box drawing of a graph $G$ with $n$ vertices drawn on an $O(n)\times O(n)$ grid
with $O(1)$ bends per edge.
Then there exists
a planarity-preserving linear morph sequence of length $O(1)$ from $P$ to 
an orthogonal box drawing $P^*$ that has $2$-spaced square boxes,
and 
the same port alignment and edge spiralities as $P$.
Moreover, $P^*$ and all explicit intermediate drawings 
in the linear morph sequence  
are drawn on an $O(n)\times O(n)$ grid
and have at most $O(1)$ bends per edge,
and the linear morph sequence can be computed in $O(n)$ time.
\end{lemma}

\begin{proof}
We will complete this morph in two steps.
First, we will morph to a drawing $P'$ that has $3$-spaced boxes that each have side-length at least $\deg(v)+2$.
Then we will use the outermost annuli of $P'$ (that is, the region contained in the $3$-proximal region but not the $2$-proximal region)
to add extra bends,
allowing us to make the boxes square, forming $P^*$.

The first step is straightforward and can be completed with two linear morphs, one horizontal and one vertical.
This is accomplished simply by adding empty $x$ and $y$ coordinates before and after
each vertex box side.
Sorting by coordinates is required, but there are $O(n)$ valid $x$ and $y$ coordinates due to the grid size of the input,
so this can be done in linear time with counting sort.
See \cref{fig:square-spaced-step1}.
The resulting drawing $P'$ and the explicit intermediate drawing meet all of our constraints for bends per edge (none added) and grid size
(a constant factor increase from $P$, plus an $O(n)$ term for the possibly increased widths/heights).

The second step also involves two linear morphs, which are once again a horizontal and vertical linear morph.
Each vertex box of a vertex $v$ has height $h\geq \deg(v)+2$ and width $w\geq \deg(v)+2$.
To perform the horizontal morph,
we add zig-zags to each edge along the top and bottom of each vertex box.
Specifically, these zig-zags are contained in the 
outermost
annulus,
so they have distance at least $\deg(v)$ from the vertex box
(either above or below it),
with each being closer to the box if it is further to the right.
The linear morph itself morphs the right side of the vertex box leftward,
so that 
its final 
distance from the left side is exactly $\deg(v)+2$. 
This is performed on all vertex boxes simultaneously,
and then repeated in a symmetric manner for the vertical linear morph.
See \cref{fig:square-spaced-step2}.
The grid size does not change, and exactly $4$ bends are added per edge.
The port alignment does not change, and added bends form zig-zags or degenerate bends and so do not affect spirality.  
Thus the resulting drawing and explicit intermediate drawing meet our constraints.
Since each vertex box only gets smaller during this step, $2$-spaced boxes are achieved in the final drawing.
The final box of a vertex $v$ has width and height exactly $\deg(v)+2$, so is a square as desired.
\end{proof}

\begin{figure}
    \centering
\hspace*{\fill}
    \includegraphics[page=1,scale=0.50]{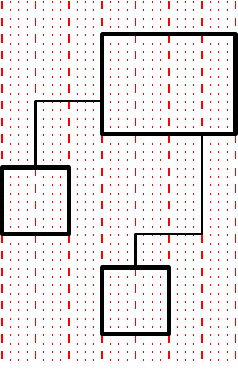}
\hspace*{\fill}
    \includegraphics[page=7,scale=0.50]{figure-add-padding-example-2}
\hspace*{\fill}
    \includegraphics[page=8,scale=0.50]{figure-add-padding-example-2}
\hspace*{\fill}
    \caption{The first step of \cref{lemma:square-spaced}.  %
    }
    \label{fig:square-spaced-step1}
\end{figure}

\begin{figure}
    \centering
\hspace*{\fill}
    \includegraphics[page=2,scale=0.25]{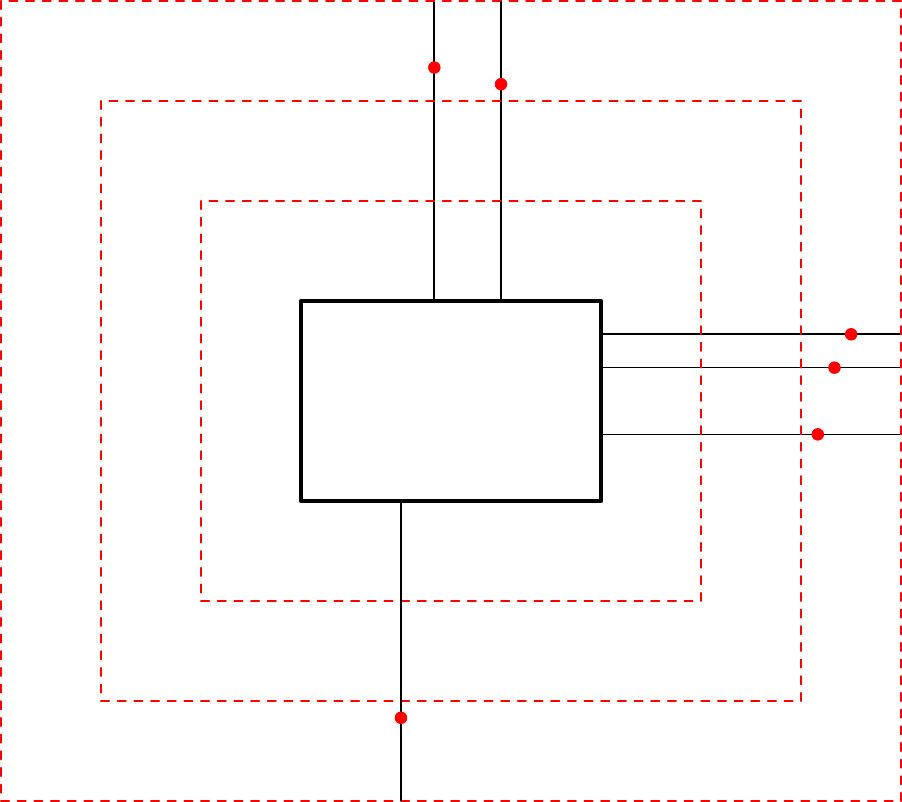}
\hspace*{\fill}
    \includegraphics[page=3,scale=0.25]{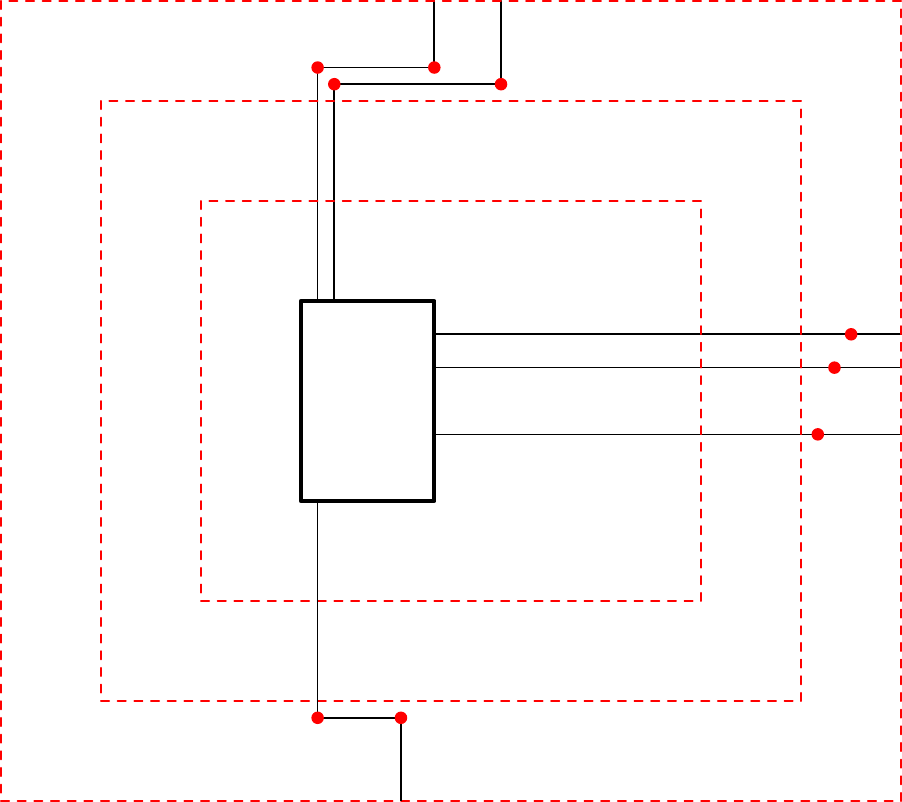}
\hspace*{\fill}
    \includegraphics[page=4,scale=0.25]{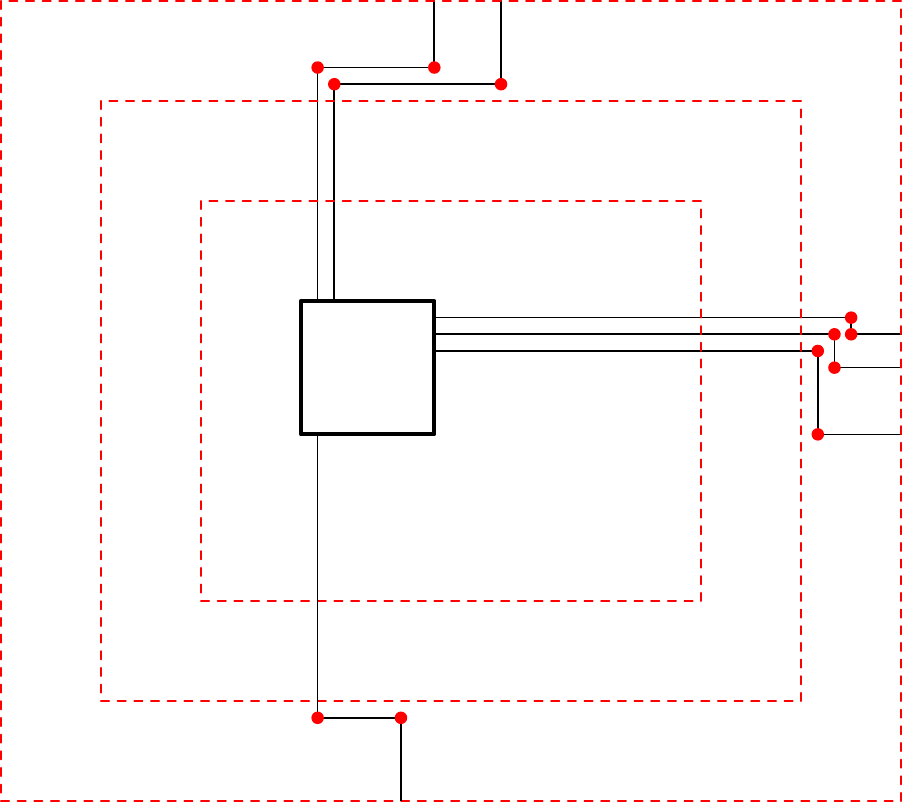}
\hspace*{\fill}
    \caption{The second step of \cref{lemma:square-spaced}.}
    \label{fig:square-spaced-step2}
\end{figure}

\begin{proof}[Proof of Lemma~\ref{lemma:planar-twist}]
To complete the proof, 
we %
show 
that there are no intersections between edges incident to different sides of the square vertex box $S$ by arguing that they live in disjoint subregions 
throughout the morph. 
See \cref{fig:twist-rotation-construction-region}.
Consider the edges incident to the top side of $S$ between top left corner $c_1$ and top right corner $c_2$. 
Their subregion consists of two parts: the rectangle  of  points inside the 2-proximal region and above the 1-proximal region, which does not change during the morph; and the points inside the 1-proximal region bounded by a ray going up from $c_1$ and a ray going right from $c_2$---this region  changes as the corners move. 
Subregions corresponding to the other three sides of $S$ are defined analogously.  These subregions are disjoint 
(note that after Phase Ia (port-alignment) we never allow ports at corners, so there are no issues at the boundaries of the subregions).
It is easy to see
that edges remain in their subregion,  
which implies that they cannot intersect.  
\end{proof}

\section{Details for Section~\ref{sec:planning-twists}: Planning Twists}
\label{appendix:planning-twists}

\begin{lemma}
\label{lemma:twist-equations} In the context of the computation of the number of twists $t(v)$ to perform at each vertex $v$ (see Section~\ref{sec:planning-twists}): 
If values $t(v)$ satisfy Equation
(\ref{eqn:vertex-twist-condition}) 
for a single vertex $v_0$, and satisfy Equation (\ref{eqn:edge-twist-condition})
for the edges of a spanning tree $T$ rooted at $v_0$, then the equations are satisfied for all vertices and edges. 
\end{lemma}
\begin{proof}
We first show that Equation (\ref{eqn:vertex-twist-condition})
is satisfied for every vertex $v$ by induction on the distance from $v_0$ to $v$ in $T$.  By assumption this holds for distance 0 (i.e., $v=v_0$), so consider the case $v\neq v_0$ and let $u$ be the neighbour of $v$ on the path to $v_0$ in $T$.  
By induction $t(u)\equiv 0 \bmod 4$.  The edge $uv$ satisfies Equation (\ref{eqn:edge-twist-condition}), so $\Delta s_P(u,v) = t(u) - t(v)$. 
Observe that because $P$ is port-aligned with $Q$, $\Delta s_P(u,v) \equiv 0 \bmod 4$.
Thus $t(v) = t(u) - \Delta s_P(u,v)
\equiv 0 \bmod 4$.

We next show that Equation (\ref{eqn:edge-twist-condition})
is satisfied for every
edge.  Let $e$ be an edge not in $T$, and let $C= v_1, v_2, \ldots, v_k$ be the cycle containing $e$  and edges of $T$  where $e= v_k v_1$.
For drawings $P$ and $Q$, consider the simple closed orthogonal paths $\gamma_P$ and $\gamma_Q$, respectively, formed by traversing edges of $C$ and, between two consecutive edges with common endpoint $v$,
following the boundary of $v$'s vertex box counterclockwise. 
Any simple %
orthogonal cycle has 4 more left turns than right turns when traversed counterclockwise~\cite{vijayan1985rectilinear}, so
$\gamma_P$ and $\gamma_Q$ have the same spirality.
Because $P$ and $Q$ are port-aligned, the portions of $\gamma_P$ and $\gamma_Q$ on the boundaries of vertex boxes involve the same turns. Thus the spirality difference between $\gamma_P$ and $\gamma_Q$ is $0 = \sum_{i+1}^k \Delta s_P(v_i,v_{i+1}) + \Delta s_P(v_k,v_1)$.  Since edges of the tree satisfy Equation~(\ref{eqn:edge-twist-condition}), $\Delta s_P(v_i,v_{i+1}) = t(v_i) - t(v_{i+1})$ for $i=1, \ldots, k-1$, so the telescoping sum collapses to 
$0 = t(v_1)-t(v_k) + \Delta s_P(v_k,v_1)$, i.e., 
$\Delta s_P(v_1,v_k) = t(v_1)-t(v_k)$.  
\end{proof}

\begin{proof}[Proof of Claim~\ref{claim:spirality-decrease}] Suppose without loss of generality that $\Delta s_P(u,v) \ge 0$ (otherwise switch $u$ and $v$).  
Then $t(u) \ge t(v)$.  
If $t(u)$ is positive and $t(v)$ is negative, then all twists at $u$ are counterclockwise, and all twists at $v$ are clockwise, so (regardless of the particular allocation to rounds) 
each round decreases the spirality of $uv$ by 2, 1, or 0.
If $t(u)$ and $t(v)$ are both non-negative, then, because of the specific allocation to rounds, we first do $t(v)$ rounds where we twist  both $u$ and $v$ counterclockwise, leaving the spirality of $uv$ unchanged, followed by $t(u)-t(v)$ rounds where we twist $u$ counterclockwise and do not twist $u$, each decreasing the spirality of $uv$ by 1; the remaining rounds twist neither $u$ nor $v$ and leave the spirality unchanged.  
A similar argument applies if $t(u)$ and $t(v)$ are both non-positive. 
\end{proof}

\section{Details for Phase II: Morphing Parallel Orthogonal Box Drawings}
\label{appendix:Phase-II}

For Phase II,
we start with a pair of parallel orthogonal box drawings $P$ and $Q$
of a connected graph $G$
of $n$ vertices
with no degenerate bends,
each drawn on an $O(n)\times O(n)$ grid.

A pair of orthogonal straight-line (point) drawings $S(P)$ and $S(Q)$
can
be introduced which correspond to $P$ and $Q$ respectively.
These are constructed by creating a vertex for each 
defining point
in the respective orthogonal box drawings,
and making the horizontal and vertical segments of the orthogonal box drawings into edges
(see \cref{fig:straight-line-from-box-example} for an example).
Since $P$ and $Q$ are parallel and have no coinciding defining points by the assumptions,
the corresponding drawings $S(P)$ and $S(Q)$ are %
both drawings of the same (labelled) graph $G'$ %
and are also parallel.
Note also that $G'$ is connected since $G$ is connected.

\begin{figure}[ht]
    \centering
    \includegraphics[page=1,scale=0.45]{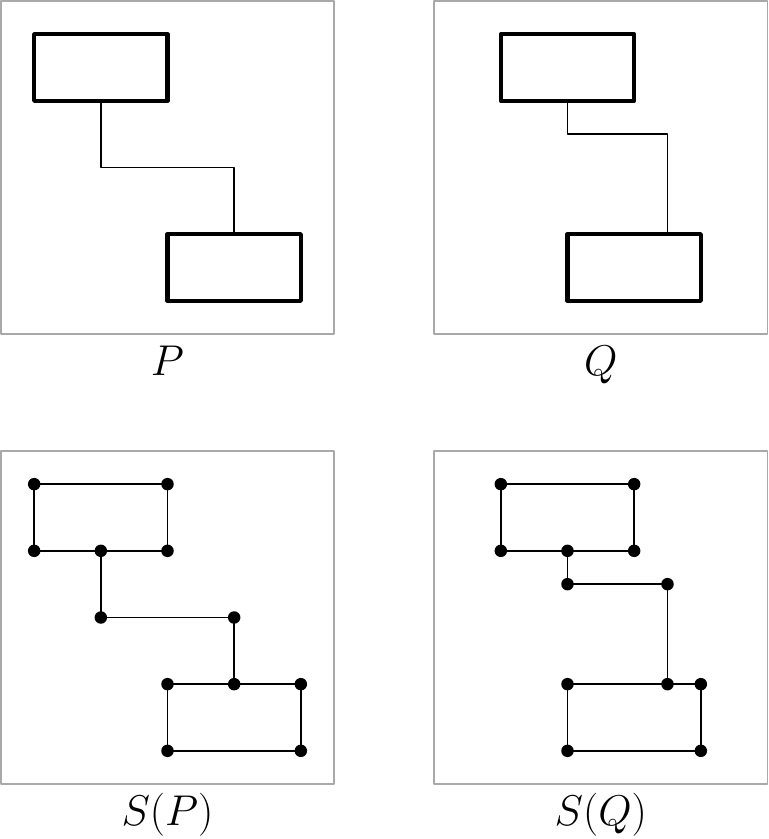}
    \caption{An example of how the straight-line orthogonal drawings $S(P)$ and $S(Q)$ are obtained from the (parallel) orthogonal box drawings $P$ and $Q$.}
    \label{fig:straight-line-from-box-example}
\end{figure}

Any sequence of planarity-preserving linear morphs from $S(P)$ to $S(Q)$ that do not change
the slopes of any edges
and do not introduce any bends
must then also induce a sequence of planarity-preserving linear morphs from $P$ to $Q$.
Biedl et al.~\cite[Theorem~5.3]{biedl2013morphing}
give
an algorithm that
finds such a sequence of $O(n)$ linear morphs
from $S(P)$ to $S(Q)$
while maintaining an $O(n)\times O(n)$ grid for explicit intermediate drawings (see the separate discussion in Section 7.1 of their paper).
In particular, they call the property that the slopes of edges do not change the ``preservation of orthogonality'' throughout the morph.
This satisfies all the necessary guarantees for Phase II, except for the time complexity.
While Biedl et al.~\cite{biedl2013morphing} do not explicitly analyze their run time,
it can be determined 
by following the steps of their algorithm
that all explicit intermediate drawings for 
our Phase II %
can be found in $O(n)$ time per drawing, hence $O(n^2)$ time overall.
We summarize Phase II as follows:

\begin{lemma}
\label{lemma:phase-ii}
Let $G$ be a connected planar graph with $n$ vertices.
Let $P$ and $Q$ be parallel orthogonal box drawings of $G$ 
on an $O(n)\times O(n)$ grid with $O(1)$ bends per edge.  
Then there is a planarity-preserving linear morph sequence
of length $O(n)$
from $P$ to $Q$
where each explicit intermediate drawing is on an $O(n)\times O(n)$ grid with $O(1)$ bends per edge.
Moreover, this sequence can be found in $O(n^2)$ time.
\end{lemma}

\end{document}